\documentclass[11pt]{article}
\usepackage{mathrsfs}
\usepackage[centertags]{amsmath}
\usepackage{amsfonts}
\usepackage{amssymb}
\usepackage{appendix} 
\usepackage{amsthm}
\usepackage{cases}
\usepackage{comment}
\usepackage{booktabs}
\usepackage{epsfig}
\usepackage {graphicx}
\usepackage{color}
\usepackage{CJK}
\usepackage{hyperref}
\usepackage{xr-hyper}
\usepackage{setspace}
\usepackage[ruled,vlined,linesnumbered]{algorithm2e}
\SetKwInput{KwIn}{Input}
\SetKwInput{KwOut}{Output}
\SetKwRepeat{Repeat}{repeat}{until}
\usepackage{subcaption}
\usepackage{natbib}
\usepackage{rotating} 
\usepackage{lscape}   
\usepackage{siunitx}              
\usepackage{geometry}             
\usepackage{multirow}
\usepackage{array}
\usepackage{tikz}
\usepackage{bm}
\usepackage{makecell}
\usepackage{soul}

\newcommand{\tabincell}[2]{\begin{tabular}{@{}#1@{}}#2\end{tabular}}

\usepackage{anysize}

\marginsize{1.1in}{1.1in}{0.55in}{0.8in} %

\newtheorem{theorem}{Theorem}[section]

\newtheorem{lemma}{Lemma}[section]

\newtheorem{assumption}{Assumption}[section]

\newtheorem{definition}{Definition}[section]
\numberwithin{equation}{section}

\def\bgamma{\boldsymbol{\gamma}}
\def\bbeta{\boldsymbol{\beta}}

\def\beps{\boldsymbol{\epsilon}}

\def\bz{\mathbf{z}}

\def\bx{\mathbf{x}}

\def\Cov{\operatorname{Cov}}
\def\Var{\operatorname{Var}}
\def\R{\mathbb{R}}

\def\bD{\mathbf{D}}

\def\Pr{\mathbb{P}}

\def\bW{\mathbf{W}}

\def\bP{\mathbf{P}}
\def\Span{\operatorname{Span}}

\begin{document}
\begin{CJK}{UTF8}{gbsn}


\title{\bf How Does LLM Help Regional CPI Forecast: An LLM-powered Deep Panel Modeling Framework}
\author{Tianchen Gao$^1$$^{\#}$, Ao Sun$^2$$^{\#}$ and Yurou Wang$^3$$^{\#}$, Jingyuan Liu$^4$$^*$ and Cheng Hsiao$^5$}
\date{}
\maketitle

\begin{center}
	\thanks{
		$^1$ Beijing International Center for Mathematical Research (BICMR), Peking University.\\
		
		$^2$ Data Sciences and Operations Department, Marshall School of Business, University of Southern California.\\
		
		$^3$ Paula and Gregory Chow Institute for Studies in Economics, Xiamen University.\\
		
		$^4$ Department of Statistics and Data Science, School of Economics, Xiamen University.\\
		
		$^5$ Department of Economics, University of Southern California.\\
		$^{\#}$These authors are equally contributed co-first authors, listed in alphabetical order. \\
		$^*$: corresponding author. E-mail: jingyuan@xmu.edu.cn
	}
\end{center}

\begin{abstract}
Understanding regional Consumer Price Index (CPI) dynamics is essential for timely and effective economic policymaking. However, traditional modeling procedures typically rely only on parametric panel modeling with low-frequency and high-cost macroeconomic indicators, which often fail to capture rapid market fluctuations and lead to inaccurate predictions. To this end, we propose a residual-joint-modeling framework that integrates large language model (LLM) analyses and social media narratives via a new deep neural network based panel modeling. Specifically, we construct a large narrative corpus from a newly collected {\it Sina Weibo} dataset, and develop a prompt-based GPT model and a series of fine-tuned BERT models to generate high-frequency LLM-induced surrogates for regional CPI. A novel joint modeling strategy is then advocated to transfer the information from these surrogates to the target regional CPI data and hence empower CPI prediction. To solve the joint objectives, we further introduce a new deep panel learning procedure with region-wise homogeneity pursuit, which has its own significance in panel data analysis literature. In addition, conformal-based panel prediction intervals are provided to quantify the uncertainty of the LLM-powered prediction. The proposed approach significantly reduces short-term forecasting errors and more effectively captures abrupt inflationary shifts compared to traditional econometric models. While demonstrated for regional CPI forecasting, the proposed framework is broadly applicable for incorporating insights from LLMs to enhance traditional statistical modeling.
\end{abstract}

\noindent\textbf{Keywords:} Large language model; Regional CPI; Panel data; Economic narrative analysis
\newpage

\section{Introduction}

Forecasting regional Consumer Price Index (CPI) is a vital yet challenging statistical and economic task. It plays an increasingly important role in policy design and economic decision-making. Local governments rely on regional CPI forecasts when setting ``CPI control targets'' and allocating fiscal transfers, price subsidies, and social-safety-net programs \citep{ciccarelli2010global,nakamura2014fiscal,cavallo2017inflation, cloyne2020monetary}. Firms similarly depend on local inflation forecasts to guide pricing, inventory management, and procurement decisions \citep{bils2004some,gorodnichenko2017price}. Moreover, regional CPI forecasting provides an empirical foundation for studying price co-movement, cost pass-through, and market integration across heterogeneous regions \citep{stock2002macroeconomic,boivin2006more}. In the existing literature, regional CPI forecasting is typically conducted using parametric panel models based on observable macroeconomic indicators, including random effects \citep{balestra1966pooling}, fixed effects \citep{mundlak1978pooling}, the Swamy estimator \citep{swamy1970efficient}, and their extensions \citep{hsiao1974statistical, hsiao1975some}. However, official CPI series and related macroeconomic indicators are often scarce, low-frequency, and costly to collect, particularly at regional level, making accurate forecasting both statistically difficult and operationally consequential \citep{pesaran2006estimation, bai2008forecasting}.

One natural direction to alleviate this informational bottleneck is to move beyond traditional macroeconomic indicators and incorporate high-frequency narrative signals. Narrative economics emphasizes that economic fluctuations are shaped not only by fundamentals but also by the diffusion of stories and collective beliefs \citep{shiller2020narrative, shiller2020popular}. Public discussions about prices and living costs on social media platforms such as {\it Sina Weibo} provide high-frequency signals of inflation expectations, which can be incorporated into econometric models in vectorized form \citep{larsen2021news, hong2025forecasting}. However, these vectorized text features are often noisy and unstable, high-dimensional embeddings introduce complex nonlinear trends, and standard panel estimators, parametric or nonparametric, struggle to integrate large-scale narrative signals while accounting for cross-regional heterogeneity. The main challenge is therefore to translate narrative diffusion into structured, interpretable, and regionally adaptive predictors within a scalable panel framework.

At the same time, Large Language Models (LLMs) provide a new opportunity to operationalize narrative measurement. By rapidly reading and summarizing massive corpora, LLMs can transform unstructured text, such as millions of time-stamped social media posts, into quantitative sentiment and belief scores for regional CPI. Therefore, it is possible to design sequential scoring prompts to obtain high-frequency inflation sentiment indices at the regional level. Yet LLM-generated predictions are inherently noisy, potentially unstable, and costly to implement; their black-box nature also raises concerns about reliability and reproducibility. Recent advances in prediction-powered and surrogate-assisted inference \citep{angelopoulos2023prediction, zrnic2024cross, mccaw2024synthetic, bashari2025synthetic, fan2025llm} highlight how machine-learning predictions can be incorporated into statistical procedures while preserving valid inference. Still, most existing frameworks focus on cross-sectional settings and impose strong assumptions on distributional similarity between target and surrogate processes. Extending these ideas to dynamic panel forecasting, where LLM-based signals serve as high-frequency surrogates for low-frequency inflation outcomes, poses a fundamental challenge: how to leverage the informational content of LLM predictions without inducing bias, instability, or model misspecification in heterogeneous regional environments.

To address these challenges, we introduce a unified LLM-powered framework for regional CPI prediction with the following \textbf{three-fold strategy}:  
\begin{enumerate}
    \item \textbf{Enriched data source:} we assemble a novel, large-scale, high-frequency, geo-tagged social-media dataset and use LLMs to extract inflation-related surrogate signals from online text. 
    \item \textbf{Joint modeling:} we integrate official low-frequency CPI data and high-frequency LLM-based surrogates within a new unified residual-joint-modeling framework, combining data enrichment and model innovation to improve forecast accuracy.  
    \item \textbf{Improved model capacity:} we develop a nonlinear, deep neural network (DNN) based panel model with homogeneity pursuit to capture complex, region-specific, and time-varying inflation dynamics.  
\end{enumerate}

First, enriched data help mitigate a fundamental bottleneck: monthly CPI is too infrequent and coarse to reflect real-time movements in local economic conditions.  
To overcome this limitation, we construct a new social media based dataset, which consists of a million-scale, multi-year panel of geo-tagged posts from {\it Sina Weibo}, one of China's largest social-media platforms. On {\it Sina Weibo}, users frequently discuss prices, cost-of-living concerns and local consumption experiences.  
Using an LLM-based extraction system, we transform millions of posts into province-level, daily surrogate indices that track inflation-related sentiment and local price perceptions. These high-frequency surrogates provide granular signals that complement sparse official CPI releases.

Second, we advocate a new residual-based joint modeling approach to combine the low-frequency CPI and high-frequency LLM-derived surrogates. We decompose prediction errors into:  
(i) a structured component explainable by surrogate residuals, and  
(ii) an idiosyncratic component.  
This decomposition reduces residual variance, stabilizes the surrogate signal, and yields more accurate point forecasts.  
Moreover, it provides the foundation for a distribution-free conformal prediction procedure, which uses these variance-reduced residuals to construct tighter and more stable prediction intervals than target-only models.

Finally, enhanced model capacity is essential because regional inflation dynamics often exhibit nonlinear interactions across time, covariates, and unobserved regional factors. A deep-learning architecture is naturally suited for capturing such complexity; a representative is the two-stage DNN-based panel approach proposed by \cite{chronopoulos2023deep}. However, excessive flexibility can be detrimental: without additional structure, a highly expressive model may overfit noise, obscure meaningful heterogeneity, and produce unstable forecasts. To this end, a key innovation of our newly developed deep panel modeling strategy is the involvement of a new homogeneity pursuit module. By introducing a classifier-Lasso \citep{su2016identifying} based  output layer during the training process, the new method identifies latent groups with similar economic responses and allows units within the same group to share parameters. This grouping strategy mitigates overfitting by reducing effective model complexity in high-dimensional settings, and enlarges the sample size available for estimating each group-level coefficient vector, thereby improving the stability and accuracy of the learned nonlinear relationships. 

All together, enriched surrogate data, residual joint modeling and nonlinear deep modeling with homogeneity pursuit form a coherent statistical framework for regional CPI forecasting. We refer to the resulting framework as \textit{LLM-powered Deep Panel Modeling (LDPM)}. It combines a deep neural network backbone with a classifier-LASSO module that induces structured sparsity and recovers latent regional clusters, maintaining interpretability while accommodating nonlinear structure. The proposed framework significantly improves provincial CPI forecasts relative to classical panel models and time-series benchmark methods. Moreover, the inferred latent regional groupings align with meaningful economic and geographic patterns, offering new insights into heterogeneous inflation dynamics across China. Although our LDPM framework is motivated by regional CPI prediction, it is directly applicable to other tasks that require prediction intervals in panel settings with scarce target data but rich surrogate information.

The remainder of the paper is organized as follows.
Section~\ref{Sec.LLM-CPI} describes the original macroeconomic data and external social media text data, together with the usage of LLMs to extract economic insight from the text data. Section~\ref{Sec.LLM-panel} introduces the proposed LDPM forecasting framework.
Section~\ref{Sec.empirical} presents the empirical analysis of provincial CPI in China and discusses the associated policy implications.
Section~\ref{Sec.simulation} reports results from numerical simulation studies.
Section~\ref{sec.discu} concludes the paper. Additional details of real data analysis, the identified and optimization of LDPM, the theoretical guarantees of conformal prediction intervals can be  found in the Supplementary Material.

\section{Data Description and Economic Insight of LLM} \label{Sec.LLM-CPI}

To forecast the regional CPI, we collect the official monthly regional CPI series released by the National Bureau of Statistics of China, titled ``Consumer Price Index (Previous Month = 100)''. To ensure comparability across regions, we standardize the CPI data by subtracting 100 from each CPI observation and dividing the result by the cross-sectional standard deviation of CPI series of each region. We also have access to two macroeconomic variables, regional Gross Domestic Product (GDP) and unemployment rate.

In addition, we obtain a high-frequency, province-level online dataset from {\it Sina Weibo} (\href{https://www.weibo.com}{weibo.com}), one of the largest social media platforms in China. This dataset captures real-time public discourse related to consumer price movements and inflation sentiment across regions. We first describe the data collecting and processing pipelines that assemble a large-scale corpus of price-related posts. Next, we develop a prompt-based LLM procedure, together with three specialized fine-tuned BERT, to identify the truly inflation-related narratives from the collected {\it Weibo} data, and to construct inflation sentiment scores that serve as informative surrogates of regional CPI. Finally, we embed the relevant text by mapping the relevant unstructured posts into numerical semantic representations for subsequent predictive modeling. For a preview, Figure~\ref{fig:workflow_overview} summarizes the entire workflow, including the data collection and preprocessing, high-frequency panel inflation index construction, and text embedding.

\begin{figure}
    \centering
    \includegraphics[width=1\linewidth]{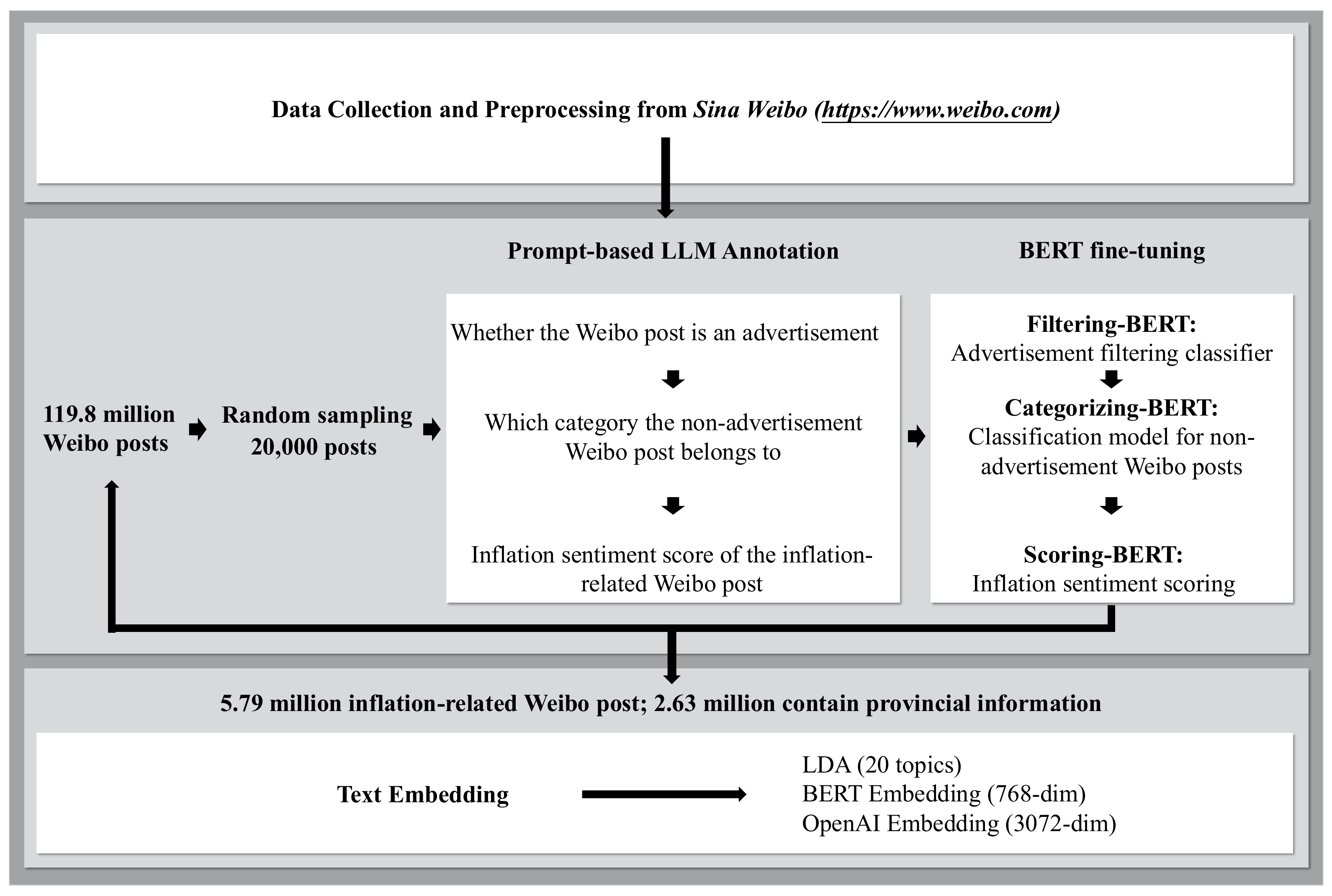}
    \caption{The overall workflow of data collection and preprocessing, high-frequency panel inflation index, and text embedding construction}
    \label{fig:workflow_overview}
\end{figure}

\subsection{Sina Weibo Corpus and Preprocessing} \label{sec:weibo_data}
{\it Sina Weibo} is one of the largest social media platforms in China and functions as a real-time information channel for both consumers and journalists. As of March 2024, the platform reports 588 million monthly active users. Weibo hosts extensive discussions spanning politics, technology, and economic affairs, and its role in shaping public discourse, particularly in financial and macroeconomic contexts, has been widely documented in the literature \citep{feng2019top, qin2024social}. Analogous to {\it Twitter}, posts on Weibo are immediately visible to followers and can be reposted, generating rich temporal and spatial diffusion patterns across regions. Because users frequently share personal experiences and subjective assessments related to daily prices, housing costs, and overall living expenses, Weibo offers a unique high-frequency lens into consumer perceptions of inflation. These publicly accessible posts, retrievable through the platform's built-in search engine, constitute a valuable data source for constructing province-level, time-indexed panel datasets that capture grassroots economic sentiment with temporal granularity comparable to, and often exceeding, that of traditional survey-based measures.

To identify discussions reflecting consumer price experiences, we adopt a keyword-based web-scraping strategy following \cite{angelico2022can} and adapt it to the polysemous nature of the Chinese language. The keyword dictionary contains 25 terms grouped into three semantic categories: (i) general price-related expressions (e.g., ``price", ``cost", ``rent"), (ii) inflation-related terms (e.g., ``price rise", ``expensive", ``inflation"), and (iii) deflation-related terms (e.g., ``price drop",  ``cheap", ``deflation"). Given the central role of housing costs in shaping consumption and saving behavior in China, we additionally incorporate housing-related keywords as complementary indicators of price perceptions. A detailed list of the keywords with their Chinese-English translations is provided in Table~\ref{tab:keyword}. We then collect all publicly available posts containing at least one of the selected keywords from January 1, 2019, to December 31, 2023. For detailed information on the data preprocessing process, please refer to the Supplementary Material \ref{app:data}. 

\begin{table}[!ht]
\centering
\scriptsize
\caption{Chinese-English keyword dictionary for identifying price-related discussions on {\it Sina Weibo}.}
    \label{tab:keyword}
    \begin{tabular}{cccccc}
        \toprule
        \textbf{Category} & \textbf{Chinese} & \textbf{English Translation} & \textbf{Category} & \textbf{Chinese} & \textbf{English Translation} \\
        \midrule
        \multirow{13}{*}{Housing} 
        & 租房 & Renting a house & \multirow{5}{*}{General price} & 价格 & Price  \\
        & 买房 & House buying & & 钱 & Money \\
        & 卖房 & House selling & & 成本 & Cost \\
        & 房价 & House price & & 费用 & Fee  \\
        & 二手房 & Used house &  & 油价 & Oil price \\
        \cmidrule{4-6}
        & 租金 & Rental fee & \multirow{4}{*}{Inflation} & 通货膨胀 & Inflation \\
        & 新房 & New house & & 涨价 & Price rise \\
        & 房屋价格 & House price & & 贵 & Expensive \\
        & 房租 & Housing rent & & 涨 & Rise (increase) \\
        \cmidrule{4-6}
        & 房贷 & Mortgage & \multirow{4}{*}{Deflation} & 通货紧缩 & Deflation \\
        & 房地产 & Real estate & & 降价 & Price reduction \\
        & 楼市 & Property market & & 便宜 & Cheap \\
        & & & & 跌 & Decline (decrease) \\
        \bottomrule
    \end{tabular}
\end{table}

\subsection{Surrogate Inflation Score Construction by LLMs} \label{textindexconst}

Although a substantial fraction of {\it Sina Weibo} posts contains information relevant to consumer price movements, a large share of the raw corpus consists of advertisements, e-commerce promotions, and other commercial or entertainment-oriented content that is unrelated to inflation dynamics. These noise are difficult to remove using conventional text-mining approaches, such as dictionary-based filters, vector space models, TF-IDF representations, or unsupervised topic models such as Latent Dirichlet Allocation (LDA), which typically lack the semantic precision required to distinguish contextual price discussions from promotional language \citep{angelico2022can}. Moreover, identifying inflation-related semantic content in social media text is intrinsically challenging due to overlapping topical categories, informal writing styles, and the heterogeneous ways in which users express opinions and personal experiences. These challenges are compounded by severe class imbalance: posts that meaningfully reflect inflation narratives constitute only a small subset of the overall corpus, rendering direct supervised classification approaches ineffective or unstable. To tackle these challenges, we are motivated to utilize LLMs as semantic filters to identify and retain inflation-relevant posts via carefully designed prompts, constructing a high-quality textual corpus that is suitable for subsequent sentiment quantification and inflation forecasting. However, solely relying on LLMs for annotation is computationally and financially infeasible at this scale of text data.

To this end, we design a prompt-based LLM annotation procedure via Generative Pre-trained Transformer (GPT), followed by a series of custom fine-tuned Bidirectional Encoder Representations from Transformers (BERT) models \citep{devlin2019bert}, to extract the semantic insight of social media text. GPT provides high-quality semantic annotations through its strong contextual understanding and few-shot learning capabilities, producing labels that are highly consistent with expert judgment \citep{gilardi2023chatgpt, korinek2023language}. These GPT-assisted labels are then used to fine-tuned BERT models, which leverage transformer-based representations to capture contextual dependencies and enable accurate classification of inflation-related narratives with much less computational and financial costs. 

To be specific, starting from the full corpus of 119.8 million {\it Sina Weibo} posts, we randomly sample 20{,}000 posts for prompt-based GPT annotation with manual inspection. The annotation process consists of three sequential tasks: identifying advertisements, classifying non-advertisement posts by content category, and assigning inflation sentiment scores to inflation-related posts. The designed sequential prompts used to train GPT are provided in the supplementary material. These annotated samples are treated as labeled training data and are used to fit three specialized BERT models: a Filtering-BERT for filtering promotional content, a Categorizing-BERT for identifying truly inflation-related posts, and a Scoring-BERT for predicting the inflation sentiment score for the identified posts. Together, these three BERT models are then applied to the original 119.8 million posts for identifying the truly relevant posts and generating the associated inflation sentiment scores. The overall workflow is summarized in Figure~\ref{fig:workflow_overview}.

This framework achieves strong and stable performance across all validation datasets. The Filtering-BERT attains an average precision (AP) of 0.97 and an area under curve (AUC) of 0.92 in identifying non-advertisement content. The Categorizing-BERT achieves an AP of 0.86 and an AUC of 0.97 in detecting inflation-related categories.
For the continuous Scoring-BERT model, its prediction mean squared error (PMSE) is 0.056.
Taking Example 1 and 2 in Table~\ref{tab:bert_example} for instance, they are first identified as non-advertisement content via Filtering-BERT, then classified as inflation-related via Categorizing-BERT, and finally assigned a high and low inflation sentiment score, 0.9 and 0.2, respectively. This tiered framework progressively removes irrelevant material, mitigates class imbalance, and preserves inflation-relevant content with high precision. 

After applying the multi-stage semantic filtering and sentiment annotation procedure, we obtain a clean corpus of approximately 5.79 million inflation-related posts. Table~\ref{tab:weibo_summary} provides an overview of the {\it Sina Weibo} dataset from January 1, 2019 to December 31, 2023. From approximately 120 million raw posts, our multi-stage LLM-based cleaning and classification pipeline identifies about 5.79 million inflation-related narratives. Additional detailed cleaning procedures and annual breakdowns are reported in the Supplementary Material.
We then extract and standardize geographic references within each post to associate textual discussions with specific regions, resulting in 2.63 million inflation-related posts spanning all 31 provinces, municipalities, and autonomous regions of China. This province-time panel dataset serves as the basis for subsequent modeling and embedding. Figure~\ref{fig:volume_inflation} plots the daily volume of inflation-related narratives on Weibo from 2019 to 2023, illustrating the intensity of public discussions about prices and inflation at a high temporal frequency. Such high-frequency text data tend to offer critical information on inflation that is overlooked by the standard monthly released official CPI indexes.

\begin{table}[!ht]
\centering
\scriptsize
\caption{Examples of the three-stage BERT workflow from two {\it Sina Weibo} posts.}
\label{tab:bert_example}
\begin{tabular}{lll}
\toprule
\textbf{Weibo Post Attributes} & \textbf{Example 1 (Inflation)} & \textbf{Example 2 (Deflation)}\\
\midrule
\textbf{Weibo Text} & \tabincell{l}{``今年药品价格涨太多了甚至\\有的成本翻了两倍不止。''} & \tabincell{l}{``猪肉降价了，赶紧整点排骨，回家给孩子\\咕嘟一锅，老香了。''} \\
\textbf{English Translation}  & \tabincell{l}{``Drug prices are crazy this \\year - some costs have more \\than doubled.''} & \tabincell{l}{``The pork price drops, so I grab some ribs,\\ head home, and stew a big pot for the kid -\\smells absolutely amazing.''}\\
\textbf{IP Location} & Beijing & Jilin \\
\textbf{Timestamp} & 2023-11-12 14:35:27 & 2022-08-04 10:28:21 \\
\midrule
\textbf{Three-stage BERT}\\
\midrule
\textbf{Stage 1: Advertisement-BERT} & Non-advertisement & Non-advertisement\\
\textbf{Stage 2: Category-BERT} & Inflation-related & Inflation-related\\
\textbf{Stage 3: CPI-BERT} & Inflation sentiment score = 0.9 & Inflation sentiment score = 0.2\\
\bottomrule
\end{tabular}
\end{table}

\begin{figure}[!ht]
    \centering
    \includegraphics[width=1\linewidth]{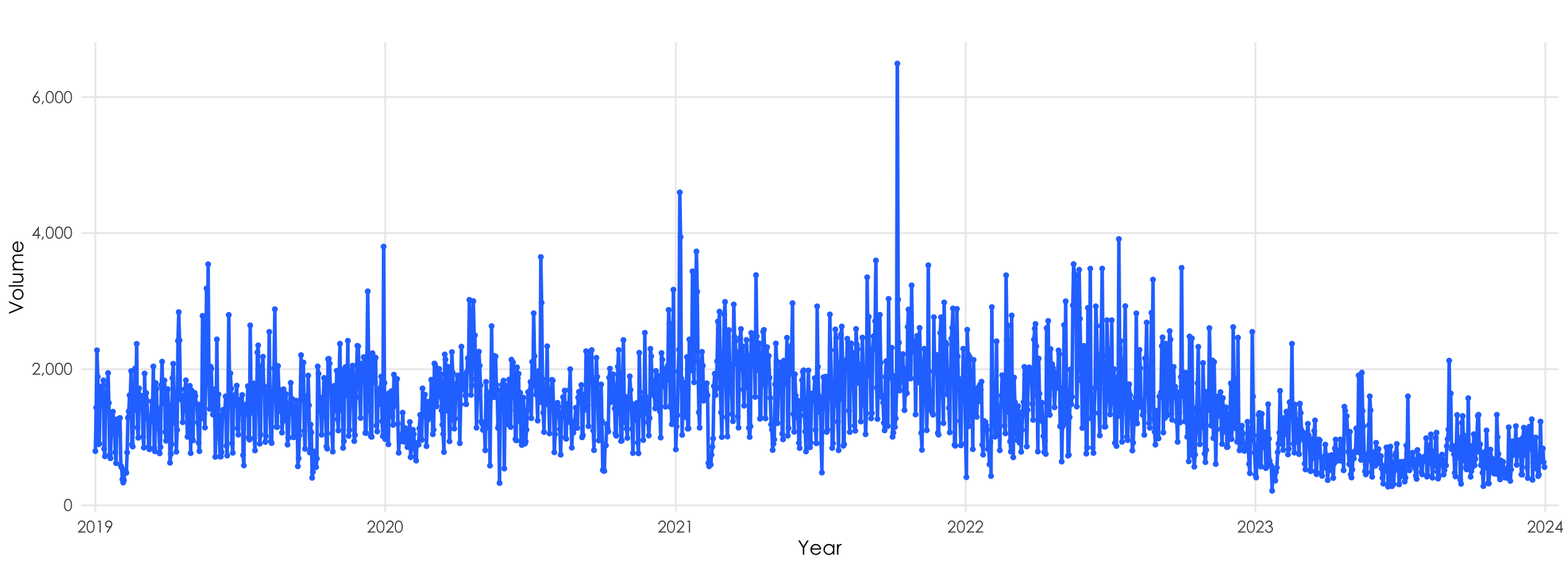}
    \caption{Daily truly inflation-related narrative volume on {\it Sina Weibo} from 2019 to 2023.
}
    \label{fig:volume_inflation}
\end{figure}

In addition, Figure~\ref{fig:volume_province} presents the provincial distribution of inflation-related social media posts collected in this study, with provinces ranked by the total number of posts mentioning inflation-related topics. In total, 2,636,795 posts were collected and geo-located at the provincial level, providing a comprehensive overview of cross-provincial variation in inflation narratives across China. We observe that the inflation-related narratives are unevenly distributed across regions, with substantially higher volumes concentrated in economically developed and densely populated provinces, particularly in eastern and southeastern China. Provinces with large urban populations and higher levels of economic activity tend to generate more inflation-related discussions, reflecting both greater exposure to price fluctuations and higher social media usage intensity.
Moreover, relatively lower posting volumes are observed in western and less populous regions. This spatial heterogeneity highlights the importance of accounting for regional differences when constructing narrative-based indicators. Treating inflation narratives as homogeneous across space may obscure meaningful variation in local economic perceptions and experiences.

\begin{figure}[!ht]
    \centering
    \includegraphics[width=1\linewidth]{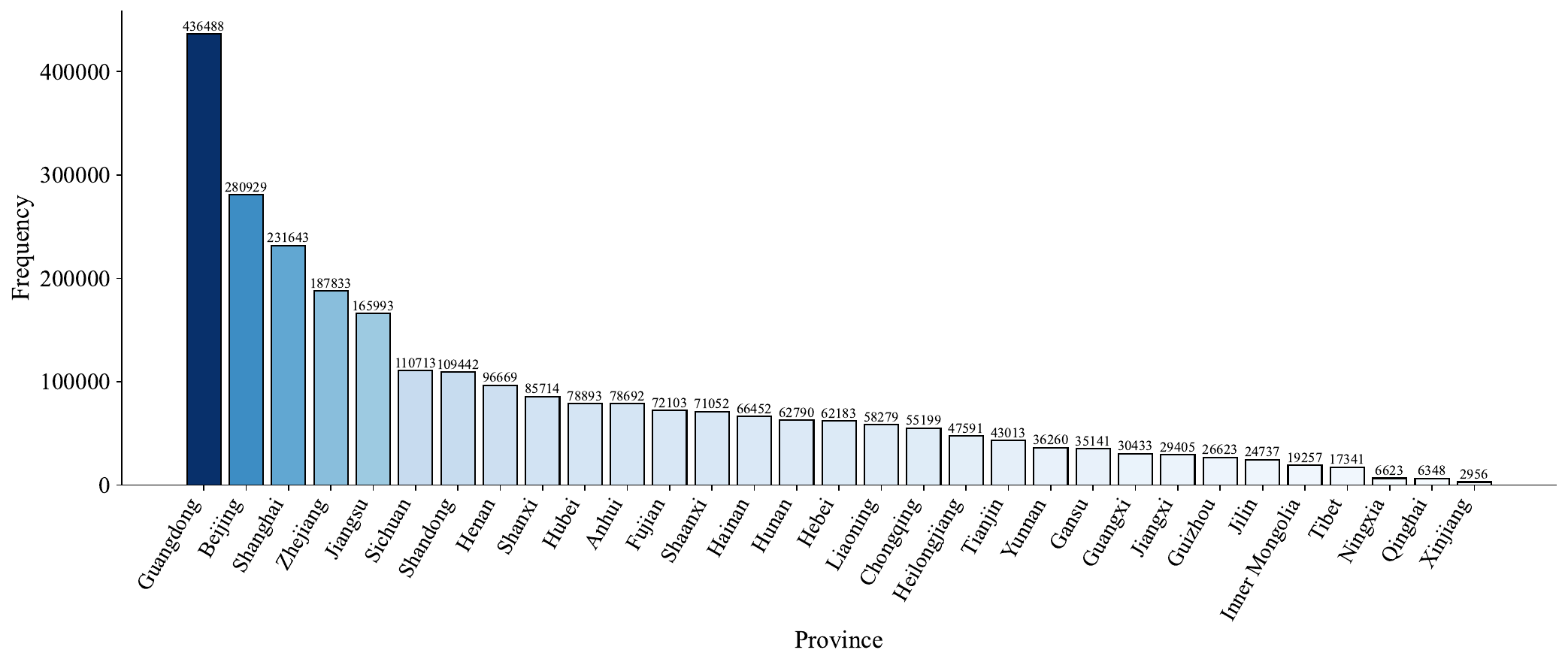}
    \caption{Provincial distribution of inflation-related posts, with provinces ranked by posting volume. The dataset contains a total of 2,636,795 posts.}
    \label{fig:volume_province}
\end{figure}

\subsection{Sina Weibo Text Embedding}\label{textembedding}

For the identified inflation-related posts, the next step is to transform these unstructured textual narratives into quantitative representations suitable for subsequent statistical modeling. 
In this study, we adopt three embedding approaches, including LDA, BERT, and generative OpenAI embeddings.
LDA embeddings represent documents as topic probability vectors, offering good interpretability and low computational cost, but they ignore word order and semantic context due to the bag-of-words assumption \citep{blei2003latent}. BERT embeddings capture contextual semantic information through a transformer architecture and generalize well across domains, though they are higher-dimensional, less interpretable, and computationally more expensive \citep{devlin2019bert}. LLM-based embeddings (e.g., OpenAI text-embedding-3-large) produce rich sentence-level representations that effectively capture complex semantics and cross-domain meanings, but they rely on external APIs and incur higher computational and monetary costs \citep{openai_embeddings}. Table~\ref{tab:weibo_example} illustrates the three types of embedded vectors from a concrete {\it Sina Weibo} post.

\begin{table}[!ht]
\centering
\scriptsize
\caption{Structure of a representative {\it Sina Weibo} record, including text content and three embedding representations.}
\label{tab:weibo_example}
\begin{tabular}{ll}
\midrule
\textbf{Weibo Text} & ``今年药品价格涨太多了甚至有的成本翻了两倍不止。'' \\
\textbf{English Translation}  & \tabincell{l}{``Drug prices are crazy this
year - some costs have more than doubled.''}\\
\midrule
\textbf{LDA Embedding} & $[0.012, 0.006, 0.000, 0.015, \dots , 0.038, 0.010
] \in \mathbb{R}^{20}$ \\
\textbf{BERT Embedding} & $[0.021, -0.134,  0.087,  0.056,  \dots  , -0.036,  0.008]  \in \mathbb{R}^{768}$ \\
\textbf{OpenAI embedding} & $[0.014, -0.072,  0.091,  0.006,  \dots  , -0.019,  0.005 ]  \in \mathbb{R}^{3072}$ \\
\bottomrule
\end{tabular}
\end{table}

\section{A Unified LLM-powered Deep Panel Modeling Framework}\label{Sec.LLM-panel}

Section \ref{Sec.LLM-CPI} introduces how to obtain the inflation insight from social media data using LLM and a series of BERT models, including the surrogate outcome of inflation sentiment scores and embedded vectors from truly inflation-related posts.
In this section, we develop a new joint modeling framework, called \textit{LLM-powered Deep Panel modeling (LDPM)}, that incorporates such generated information to empower the target CPI modeling. We first propose a surrogate augmented model via a residual dependence structure. To adapt the nonlinear trend of the high-dimensional embedding upon fitting the surrogate augmented model, we further develop a modified DNN modeling strategy with region-wise homogeneity pursuit - this has its own significance in panel data modeling literature. 
We note that the proposed framework is not limited to CPI forecasting; it is broadly applicable to panel data analysis in general. It is designed to reduce both the \emph{noise variance} and the \emph{estimator variance} in panel modeling, leading to more accurate and precise prediction. 

\subsection{Joint Modeling via Surrogate Augmentation}\label{subsec:models}

Let $y_{i,t}\in\mathbb{R}$ be the target outcome of primary interest, where $i=1,\ldots,N$ indexes cross-sectional units and $t=1,\ldots,T$ indexes time, typically low-frequency. In our study, $y_{i,t}$ refers to the monthly released regional CPI. Suppose we also observe macroeconomic covariates $\mathbf{z}_{i,t}\in\mathbb{R}^{d_z}$, such as housing price index, unemployment rate, etc. For simplicity, fixed time effects or lagged values of $y_{i,t}$ are absorbed into $\mathbf{z}_{i,t}$. $\{\mathbf{x}_{i,t,k}\in\mathbb{R}^{d_x}\}_{k=1}^{K_t}$ denote the vector embedded from the intra-month online text datasets $\{\mathcal{D}_{i,t,k}\}_{k=1}^{K_t}$, where $K_t$ is the number of posts in month $t$. The \textsc{target} model for predicting $y_{i,t}$ is
\begin{equation}\label{eq:model-general}
y_{i,t} = F_i(\{\mathbf{x}_{i,t,k}\}_{k=1}^{K_t}, \mathbf{z}_{i,t}) + \epsilon_{i,t},
\quad \mbox{where } \ \  \mathbb{E}[\epsilon_{i,t}\mid \{\mathbf{x}_{i,t,k}\}_{k=1}^{K_t},\mathbf{z}_{i,t}]=0,
\end{equation}
where $F_i(\cdot)$ is a flexible nonparametric function for the $i$th unit. 

Further denote the generated surrogate outcome as $\{y^S_{i,t,k}\}_{k=1}^{K_t}$. Then, we impose the following \textsc{surrogate} model:
\begin{equation}\label{eq:surrogate-general}
y^S_{i,t,k} = G_{i}(\mathbf{x}_{i,t,k}, \{y^S_{i,t,k-l}\}_{l=1}^{q_{k}}) + \epsilon^S_{i,t,k},
\quad \mbox{where } \ \  \mathbb{E}[\epsilon^S_{i,t,k}\mid \mathbf{x}_{i,t,k}, \{y^S_{i,t,k-l}\}_{l=1}^{q_{k}}]=0,
\end{equation}
where $G_{i}(\cdot)$ is a nonparametric function that vary across regions, and $\{y^S_{i,t,k-l}\}_{l=1}^{q_{k}}$ is the collection of $q_k$ lagged terms.

Let $\mathbf{\epsilon}^S_{i,t}=(\epsilon^S_{i,t,1},\ldots,\epsilon^S_{i,t,K})^\top$. To reduce noise in~\eqref{eq:model-general}, assume the following nonparametric dependence structure between $\epsilon_{i,t}$ and $\mathbf{\epsilon}^S_{i,t}$:
\begin{equation}\label{eq:proj-gamma}
\epsilon_{i,t} = \Gamma(\mathbf{\epsilon}^S_{i,t}) + e_{i,t}, \quad \mbox{where } \ \  \mathbb{E}[e_{i,t}\mid \mathbf{\epsilon}^S_{i,t}]=0.
\end{equation}

Combining \eqref{eq:model-general}-\eqref{eq:proj-gamma} yields our \textsc{surrogate augmented} model
\begin{equation}\label{eq:joint-Q}
y_{i,t} = Q_i(\{\mathbf{x}_{i,t,k}\}_{k=1}^{K_t},\mathbf{z}_{i,t},\mathbf{\epsilon}^S_{i,t}) + e_{i,t},
\quad \mbox{where } \ \ Q_i(\mathbf{x},\mathbf{z},\mathbf{\epsilon}^S)=F_i(\mathbf{x},\mathbf{z})+\Gamma(\mathbf{\epsilon}^S).
\end{equation}

This formulation integrates both the official economic predictors and the fine-grained surrogate information extracted from LLMs and social media via residual dependence between the target and surrogate models. 

We use the error dependence structure to link the target and surrogate models, rather than directly adding $y^S_{i,t,k}$ to the target model. This is because directly including $y^S_{i,t,k}$, despite its predictive power for $y_{i,t}$, can induce so-called shortcut learning \citep{geirhos2020shortcut}, leading to a collapse in learning the structural response function $F_i(\cdot)$ in the target model and contaminates the homogeneity pursuit step \citep{papyan2020prevalence}. See detailed discussion in Section~\ref{app:shortcut} of the Supplementary Material. By contrast, the surrogate residual $\epsilon^S_{i,t,k}$ captures the \emph{unexplained component} of the surrogate outcome - the portion of latent wisdom from LLM pre-training which is not directly accounted for by observed embeddings. Using surrogate errors therefore provides a cleaner and more targeted auxiliary signal for improving prediction. Additionally, $y^S_{i,t,k}$ often exhibits substantial variance and instability across
different surrogate-model training runs. After removing the mean structure, the error left have markedly smaller variance and are consequently more stable inputs for
predictive inference. 

\subsection{Deep Panel Training: DNN with Hidden Homogeneity Pursuit}\label{sec:deeppanel}

Due to the curse of dimensionality, classical nonparametric methods are impractical for estimating $Q_i(\cdot)$ in the surrogate augmented model \eqref{eq:joint-Q} using the high-dimensional embedded vector $\bx_{i,t,k}$. In such settings, deep neural networks (DNN) offer a powerful solution \citep{schmidt2020nonparametric}. In practice, however, the low-frequency economic and macro-financial datasets are typically not large enough to support fully heterogeneous region-wise DNNs. To address this, we develop a modified DNN training procedure, called {\it Deep Panel Training}, in this section. It incorporates latent group structure by employing a shared feature-extraction network while allowing for heterogeneity through a group-specific final linear output layer. The proposed training procedure itself contributes to the panel modeling literature by bridging the complex deep learning techniques with latent group structures in panel data. 

Specifically, we approximate the surrogate augmented model \eqref{eq:joint-Q} using a shared multi-layer nonlinear feature map together with unit-specific linear heads and the surrogate correction:
\begin{equation}\label{eq:deep-panel-model}
Q^{\mathrm{DNN}}_i(\bx,\bz,\beps^S)
=
\bbeta_i^\top
h\left( \{\bx_{i,t,k}\}_{k=1}^{K_t}, \bz_{i,t}, \beps^S_{i,t}; \bW,\bgamma \right) + b_i,
\end{equation}
where the shared feature map $h(\cdot)$ is generated by an $L$-layer feedforward network,
\begin{equation}\label{h-eq}
h(\boldsymbol\nu;\bW,\bgamma)
=
\sigma\Bigl(
\bW^{(L)}
\sigma\bigl(
\bW^{(L-1)}\cdots
\sigma(\bW^{(1)}\boldsymbol\nu+\bgamma^{(1)})
+\cdots+\bgamma^{(L-1)}
\bigr)
+
\bgamma^{(L)}
\Bigr),
\end{equation}
with activation function $\sigma(\cdot)$, such as ReLU or Sigmoid.  The parameters $\bW=\{\bW^{(l)}\}_{l=1}^L$ and $\bgamma=\{\bgamma^{(l)}\}_{l=1}^L$ are shared across all units, while $\bbeta_i\in\R^{d_h}$ denotes the unit-specific output coefficients.

To further borrow strength across units, we impose a latent grouping structure on the output heads. Specifically, suppose there exist $K_0$ latent centers $\{\boldsymbol\eta_k\}_{k=1}^{K_0}$,  $\{\varphi_k\}_{k=1}^{K_0}$ ,and a group assignment function $g:\{1,\ldots,N\}\to\{1,\ldots,K_0\}$ such that
\[
\bbeta_i = \boldsymbol\eta_{g(i)}, \quad  b_i = \varphi_{g(i)}.
\]
The latent group structure permits heterogeneity while enabling the estimation of common patterns shared within each latent group.

We emphasize that the parameters $(\bW,\bgamma)$ are deliberately shared across all units, while only the final-layer coefficients $(\bbeta_i,b_i)$ vary by region. This design guarantees both stability and identifiability. First, when the number of regions $N$ is large but the per-unit sample size is limited, as is the case for our monthly macroeconomic dataset, fully unit-specific DNN would result in an explosion of parameters and substantial under-training. A shared feature extractor drastically reduces dimensionality and stabilizes estimation. Second, parameter sharing ensures identifiability of the unit-level coefficients $\{(\bbeta_i,b_i)\}_{i=1}^N$ up to label permutations. It is well known that the internal parameters of DNN are typically non-identifiable: even if the underlying group structure is truly present, there may exist another DNNs with parameters $(\tilde{\bW},\tilde{\bgamma},\{(\tilde{\bbeta}_i,\tilde{b}_i)\}_{i=1}^N)$ such that
\[
\tilde{Q}^{\mathrm{DNN}}_i(\bx,\bz,\beps^S)
=
\tilde{\bbeta}_i^\top 
h\left(\{\bx_{i,t,k}\}_{k=1}^{K_t}, \bz_{i,t}, \beps^S_{i,t}; \tilde{\bW},\tilde{\bgamma}\right) + \tilde{b}_i
=
Q^{\mathrm{DNN}}_i(\bx,\bz,\beps^S), i=1,\ldots,N.
\]

In this case, the coefficient vectors $(\bbeta_i,b_i)$ are not uniquely determined, and imposing a group structure would be meaningless. However, with shared parameters $(\bW,\bgamma)$, the equivalence class of DNNs is sharply restricted. Under mild conditions on the activation and architecture in Section~{\ref{app:identifiability}} of the Supplementary Material, any two DNNs implementing the same feature map $h(\cdot)$ must be related only through rotations and scalings of the hidden representation. Concretely, if both parameter sets represent the same function, then there exist an permutation matrix $\bP$ and a positive diagonal scaling matrix $\mathrm{diag}(\Lambda)$ such that
$$
\begin{bmatrix}
(\tilde{\bbeta}_1^\top,\tilde{b}_1) \\
\vdots \\
(\tilde{\bbeta}_N^\top,\tilde{b}_N)
\end{bmatrix}
=
\bP\mathrm{diag}(\Lambda)
\begin{bmatrix}
(\bbeta_1^\top,b_1)\\
\vdots \\
(\bbeta_N^\top,b_N)
\end{bmatrix}.
$$

The shared-parameter architecture substantially restricts the admissible reparametrizations of the network. In particular, if two parameter sets generate the same shared feature map, they can differ only through a common permutation and coordinate-wise scaling of the hidden representation. Accordingly, the collection of output heads $\{(\bbeta_i,b_i)\}_{i=1}^N$ is identified only up to the corresponding common scaling and label permutation. In contrast, without the shared-parameter architecture, each unit-specific DNN may admit arbitrary reparametrizations, rendering $\{(\bbeta_i,b_i)\}_{i=1}^N$ fundamentally unidentifiable and making the latent group structure ill posed. See Section \ref{app:identifiability} of the Supplementary Material for detailed discussion.

\subsection{Two-stage Fitting Algorithm for Deep Panel Training}\label{subsec:estimation}

In this section, we provide the two-stage algorithm that realizes the deep panel training described in Section \ref{sec:deeppanel} for the proposed surrogate augmented model \eqref{eq:joint-Q}.
A central component of our estimation strategy is the recovery of the latent group structure.  
To achieve this, we extend the idea of classifier-LASSO (C-LASSO) \citep{su2016identifying} to the DNN setting by imposing the penalty on the unit-specific output heads of the network. Building on this, the entire estimation procedure consists of the following two stages.

\ \\
\noindent\textbf{Stage 1: Surrogate model fitting.}

We first estimate $G_{i}$ in the surrogate model \eqref{eq:surrogate-general} via regular DNN with all data from region $i$. Then construct surrogate residuals
\[
\widehat{\epsilon}^S_{i,t,k}
=
y^S_{i,t,k} - \widehat{G}_{i}(\bx_{i,t,k},  \{y^S_{i,t,k-l}\}_{l=1}^{q_{k}}),
\qquad
\widehat{\beps}^S_{i,t}
=
(\widehat{\epsilon}^S_{i,1},\ldots,\widehat{\epsilon}^S_{i,t})^\top.
\]

We remark that for estimating $G_{i}$, we may also adopt the same Deep Panel Training with hidden homogeneity pursuit described in Section \ref{sec:deeppanel}. However, the primary fitting objective here is to derive an accurate proxy for the surrogate error $\boldsymbol\epsilon^S$, which is subsequently used in fitting the surrogate augmented model \eqref{eq:joint-Q}. Therefore, provided that the sample size $K_t$ is usually sufficient for separate DNN fittings per region, homogeneity pursuit is not required. 

\ \\
\noindent\textbf{Stage 2: Surrogate augmented model fitting.}

Using $\widehat{\beps}^S_{i,t}$ from Stage~1, we fit the surrogate augmented model \eqref{eq:joint-Q} via Deep Panel Training in two steps. We first use a penalized objective to recover the latent group structure, and then refit the group-specific output heads without penalty to obtain the final coefficient estimates. Let $h_{i,t}=h(\{\bx_{i,t,k}\}_{k=1}^{K_t},\bz_{i,t},\widehat{\beps}^S_{i,t};\bW,\bgamma)$ denote the shared DNN feature mapping in \eqref{h-eq}. Since the shared hidden representation is identified only up to a common coordinate-wise scaling, the unit-specific output heads are only identified up to the corresponding inverse transformation. To make the grouped penalty well defined, we normalize the final hidden representation on a common scale across all units and time periods before comparing the output heads across regions. 

The homogeneity-pursuit objective for estimating $(\bbeta_i,\bW,\bgamma)$ becomes
\begin{equation}\label{equ:jan24:01}
 \min_{\{(\bbeta_i,b_i)\},\bW,\bgamma}
\frac{1}{NT}
\sum_{i=1}^N \sum_{t=1}^T
\left(
    y_{i,t} - \bbeta_i^\top h_{i,t}- b_i
\right)^2
+ \frac{\lambda}{N}
\sum_{i=1}^N
\prod_{k=1}^{K_0}
\left(\big\| (\boldsymbol\eta_k^\top,\varphi_k)^\top - (\bbeta_i^\top,b_i)^\top \big\|_2\right).   
\end{equation}

The product penalty encourages each $(\bbeta_i,b_i)$ to be close to one of the $K_0$ group centers  
$\{ \boldsymbol\eta_k, \varphi_k\}_{k=1}^{K_0}$.  
A discussion of optimization strategies for this loss is provided in Section \ref{app:optimization} of the Supplementary Material.  
This yields the estimated group assignment  
\[
\widehat{g}:\{1,\ldots,N\}\to\{1,\ldots,K_0\},
\]
together with the shared feature parameters $(\widehat\bW,\widehat\bgamma)$. Conditioning on the estimated structure $(\widehat{g},\widehat\bW,\widehat\bgamma)$, for each $k=1,\ldots,K_0$, we can further refit the group centers by solving
\[
(\widehat{\boldsymbol\eta}_k,\hat{\varphi}_k)
=
\arg\min_{\boldsymbol\eta, \varphi}
\sum_{i:\widehat g(i)=k}
\sum_{t=1}^T
\bigl(
    y_{i,t} - \boldsymbol\eta^\top \widehat{h}_{i,t}-\varphi
\bigr)^2,
\qquad
\widehat{\bbeta}_i := \widehat{\boldsymbol\eta}_{\widehat g(i)}, \hat{b}_i:= \hat{\varphi}_{\hat{g}(i)}.
\]

This final refit debiases the group-level parameters after selection  
and yields the final Deep Panel Training (DPT) estimator of $Q_i$:
\begin{equation*}
\widehat{Q}^{\mathrm{DPT}}_i(\bx,\bz,\beps^S)
=
\widehat{\boldsymbol\eta}_{\widehat g(i)}^\top
h\left(
    \{\bx_{i,t,k}\}_{k=1}^{K_t}, \bz_{i,t}, \beps^S_{i,t};
    \widehat\bW,\widehat\bgamma
\right) + \hat{\varphi}_{\hat{g}(i)},
\end{equation*}
where $\widehat g(i)$ is the estimated group assignment from Stage~2.

For the one-step-ahead prediction at time $T+1$,  
we first construct surrogate residuals using the surrogate model:
\[
\widehat{\epsilon}^S_{i,T+1,k}
=
y^S_{i,T+1,k} - \widehat{G}_{i}(\bx_{i,T+1,k},  \{y^S_{i,t,k}\}_{t=1}^{q_{k}}),
\qquad
\widehat{\beps}^S_{i,T+1}
=
(\widehat{\epsilon}^S_{i,T+1,1},\ldots,\widehat{\epsilon}^S_{i,T+1,K_{T+1}})^\top.
\]

Given the observed covariates $(\bx_{i,T+1,k},\bz_{i,T+1})$,  
the point predictor of $y_{i,T+1}$ is
\[
\widehat{y}_{i,T+1}
=
\widehat{Q}^{\mathrm{DPT}}_i
\bigl(
    \{\bx_{i,T+1,k}\}_{k=1}^{K_{T+1}},
    \bz_{i,T+1},
    \widehat{\beps}^S_{i,T+1}
\bigr).
\]

For a longer prediction horizon $h>1$,  
we use a standard rolling (recursive) forecasting scheme whenever $\bz_{i,t}$ includes lagged values of $y_{i,t}$.  
Denote by $\widehat{\bz}_{i,T+h}$ the covariate vector obtained from $\bz_{i,T+h}$ by replacing all lagged outcomes $y_{i,t}$ with their predictions $\widehat{y}_{i,t}$ for $t>T$.  
Then the $h$-step-ahead predictor is
\[
\widehat{y}_{i,T+h}
=
\widehat{Q}^{\mathrm{DPT}}_i
\bigl(
    \{\bx_{i,T+h,k}\}_{k=1}^{K_{T+h}},
    \widehat{\bz}_{i,T+h},
    \widehat{\beps}^S_{i,T+h}
\bigr),
\]
where $\widehat{\beps}^S_{i,T+h}$ is constructed analogously to $\widehat{\beps}^S_{i,T+1}$ using the surrogate model and the contemporaneous text embeddings at time $T+h$. See Section \ref{app:optimization} of the Supplementary Material for the details of computation.

\subsection{Conformal Prediction Inference}\label{subsec:conformal-PI}

Based on the point prediction obtained in Section \ref{subsec:estimation}, we further construct prediction sets using a chronologically split conformal method \citep{xu2023conformal} to depict the prediction uncertainty. 
Concretely, partition the panel data chronologically into the training set 
$ \mathcal{D}_{\text{Tr}} := \{(i,t): i=1,\ldots,N, t=1,\ldots,T_1 \}$, the validation set $\mathcal{D}_{\text{Val}} := \{(i,t): i=1,\ldots,N, t=T_1+1,\ldots,T_2\}$ and the calibration set $
\mathcal{D}_{\text{Cal}}  := \{(i,t): i=1,\ldots,N, t=T_2+1,\ldots,T \}$. We first conduct the \textit{Deep Panel Training} using $\mathcal{D}_{\text{Tr}}$, tuning hyperparameters based on   
$\mathcal{D}_{\text{Val}}$.  
After selecting the optimal hyperparameters,  
we refit the model on to obtain the final estimator $\widehat{Q}^{\mathrm{DPT}}_i$. For each $(i,t)\in\mathcal{D}_{\text{Cal}}$,  
construct the point prediction  
$\widehat{y}_{i,t}=\widehat{Q}^{\mathrm{DPT}}_i(\bx_{i,t},\bz_{i,t},\widehat{\beps}^S_{i,t})$  
and define the absolute-residual nonconformity score
$$
s_{i,t} = |y_{i,t} - \widehat{y}_{i,t}|.
$$

Since the resulting \textit{Deep Panel Training} estimator includes group labels via homogeneity pursuit, we calibrate within groups.  
Let $\widehat g(i)$ denote the estimated group membership of unit $i$.  
For a given group $k$, collect the calibration scores
$$
\mathcal{A}_k
    := \{ s_{j,t} : (j,t)\in\mathcal{D}_{\text{Cal}}, \widehat g(j)=k \}.
$$

Let $m_k := |\mathcal{A}_k|$ and $s_{k,(1)} \le \cdots \le s_{k,(m_k)}$ be the order statistics of $\mathcal{A}_k$.  
The $(1-\alpha)$100\% conformal quantile for group $k$ is
$$
q_k
=
s_{k,(\lceil (m_k+1)(1-\alpha) \rceil)}.
$$

Consider a test point $(i^\ast,t+h)$ with estimated group label  
$k^\ast = \widehat g(i^\ast)$ and point prediction
$$
\widehat{y}_{i^\ast,t+h}
=
\widehat{Q}^{\mathrm{DPT}}_{i^\ast}
    (\bx_{i^\ast,t+h},\hat \bz_{i^\ast,t+h},\widehat{\beps}^S_{i^\ast,t+h}).
$$
The $(1-\alpha)$100\% split conformal prediction set is
$$
\mathcal{C}_{i^\ast,t+h}(1-\alpha)
=
\bigl[
    \widehat{y}_{i^\ast,t+h} - q_{k^\ast},
    \widehat{y}_{i^\ast,t+h} + q_{k^\ast}
\bigr].
$$

This interval achieves marginal $(1-\alpha)$100\% coverage within each estimated group, leveraging both cross-sectional information and the latent structural homogeneity.
Provided that $\widehat{Q}^{\mathrm{DPT}}_i$ is a consistent estimator of the true regression function $Q_i$, the empirical distribution of the nonconformity score
\[
s_{i,t} = |y_{i,t} - \widehat{Q}^{\mathrm{DPT}}_i(\bx_{i,t},\bz_{i,t},\widehat{\beps}^S_{i,t})|
\]
converges to the population distribution of the residuals  
$|y_{i,t} - Q_i(\bx_{i,t},\bz_{i,t},\beps^S_{i,t})|$.  
Consequently, the empirical conformal quantile converges to the population residual quantile,  
and the resulting prediction set attains asymptotic $(1-\alpha)$100\% coverage under mild regularity conditions:
$$
\Pr\left\{y_{i^\ast, T+1} \in \mathcal{C}_{i^\ast,t+1}(1-\alpha)\right\} \to 1-\alpha.
$$

The theoretical guarantees of this conformal inference are provided in Section~\ref{sec:theoryconformal} of the Supplementary Material. There, we establish the coverage properties of the resulting prediction intervals and rigorously show that joint modeling yields shorter intervals compared to that without LLM surrogate augmentation. Intuitively, the efficiency gain of the constructed conformal set stems from two sources. First, the joint modeling strategy decomposes the noise into a \emph{structured} component explainable by surrogate residuals $\beps^S_{i,t}$ and an \emph{idiosyncratic} component $e_{i,t}$. This reduces the variance of the nonconformity score $s_{i,t}$ relative to that constructed from a target-only model. Since the conformal quantile depends on the tail of the distribution of $s_{i,t}$, smaller residual variance leads directly to a tighter prediction interval. Second, the homogeneity pursuit step recovers latent groups and pools information across units belonging to the same group.  
    By replacing $\{\bbeta_i\}$ with shared group-level coefficients  
    $\{\boldsymbol\eta_k\}$, the effective sample size for estimating each coefficient vector increases.  
    More accurate estimation of both $(\widehat\bW,\widehat\bgamma)$ and  
    the group-level heads $\{\widehat{\boldsymbol\eta}_k\}$  
    reduces prediction error, thereby shrinking the distribution of nonconformity scores  
    and yielding narrower conformal sets.

In sum, the entire procedure is unified into the so-called LDPM framework, including the LLM surrogate construction, the joint modeling strategy via surrogate augmentation, the DPT algorithm, and the conformal prediction.

\section{Regional CPI Forcasting via LLM-powered Deep Panel Modeling}\label{Sec.empirical}

\subsection{Variable Description}

In this section, we utilize the proposed LDPM strategy to forecast the monthly regional CPI. The target outcome is denoted as $y_{i,t}$, for $i = 1, \ldots, N$ and $t = 1, \ldots, T$, representing the normalized monthly inflation rate for region $i$ at time $t$. The macroeconomic variables $\bz_{i,t}$ includes the regional GDP and unemployment rate. The surrogate outcome $y_{i,t,k}^{S}$ are constructed from the filtered {\it Sina Weibo} corpus, as detailed in Section~\ref{textindexconst}, and the predictors $\bx_{i,t,k}$ are derived from textual embedding methods in Section \ref{textembedding}. 
For each region and month, the post-level embeddings are aggregated to the day-level through a commonly used max pooling strategy \citep{collobert2011natural,shen2018baseline}, and the resulting vectors are temporally aligned with the daily averages of post-level surrogate outcome scores. Note that we adopt different aggregation strategies for surrogate outcome and textual embedding. The surrogate outcome is designed as an integrated city-day measure of inflation narratives and is therefore constructed by averaging post-level signals, whereas text embeddings aim to preserve the strongest semantic activations across posts, for which max pooling retains the most salient features along each embedding dimension. Detailed descriptions of all variables are provided in Table~\ref{tab:variables}.

\begin{table}[!ht]
\centering
\scriptsize
\caption{Description of variables used in the target and surrogate models}
\label{tab:variables}
\begin{tabular}{ccl}
\toprule
\textbf{Type} & \textbf{Variable} & \textbf{Description} \\ 
\midrule
\multicolumn{3}{l}{\textbf{Target Model}}\\
\midrule
Response & $y_{i,t}$ & Official regional CPI \\[2pt]
\multirow{2}{*}{Predictors}  
& $\mathbf{x}_{i,t}$ & Aggregated post-level embeddings \\[2pt]
& $\mathbf{z}_{i,t}$ & Macroeconomic covariates (regional GDP and unemployment rate) \\[2pt]
\midrule
\multicolumn{3}{l}{\textbf{Surrogate Model}}\\
\midrule
Response & $y^S_{i,t,k}$ & Surrogate CPI index for region $i$ at month $t$, day $k$. \\[2pt]
Predictors & $\mathbf{x}_{i,t,k}$ & Aggregated post-level embeddings for region $i$ at month $t$, day $k$ \\ 
\bottomrule
\end{tabular}
\end{table}

\subsection{Out-of-sample Forecasting}\label{oos}

To assess the out-of-sample forecasting performance of our proposed LDPM framework, we specify three LDPM strategies using different types of textual embeddings: OpenAI embeddings, LDA embeddings, and BERT embeddings.
For comparison, we follow the traditional linear panel modeling approach (LPM) \citep{hsiao2022analysis} and consider a benchmark specification that includes only macroeconomic covariates $\mathbf{z}_{i,t}$ (unemployment rate and regional GDP). Additionally, we also compare models that augment the macroeconomic covariates with the three aforementioned embeddings (LPM-E) following \cite{hong2025forecasting}.

The out-of-sample forecast period spans from $H$ months prior to December 2023. The whole dataset is split temporally into training, validation, and testing (out-of-sample) periods: The validation set consists of the six months immediately preceding the testing period, used for hyperparameter fine-tuning such as the number of hidden layers, neurons, and the adaptive learning rate schedule. 

{\small
$$
\underbrace{\text{Training set}}_{\text{From Jan 2021}} 
\longrightarrow
\underbrace{\text{Validation set}}_{\text{6 months prior to testing period}} 
\longrightarrow
\underbrace{\text{Testing set (out-of-sample)}}_{\text{$H$-month period ending in Dec 2023}}
$$
}

We report the prediction mean squared errors (PMSE) in Table~\ref{tab:real_result_3tables}.

\begin{table}[!htbp]
\centering
\footnotesize
\caption{The $PMSE(H)$ results across different horizons $H$}
\label{tab:real_result_3tables}
\begin{tabular}{ccccccccccc}
\toprule
\textbf{Method} & \textbf{Embedding} & \textbf{$8$} & \textbf{$9$} & \textbf{$10$} & \textbf{$11$} & \textbf{$12$} & \textbf{$13$} & \textbf{$14$} & \textbf{$15$} & \textbf{Ave.}\\ 
\midrule
\textbf{LPM} & \textbf{N/A} & 2.387 & 2.479 & 2.411 & 2.163 & 2.246 & 2.394 & 2.390 & 2.449 & 2.365 \\
\midrule
\multirow{3}{*}{\textbf{LPM-E}} & \textbf{LDA} & 0.638 & 0.685 & 0.805 & 1.019 & 1.055 & 1.037 & 0.954 & 0.984 & 0.897 \\
& \textbf{BERT} & 0.744 & 1.024 & 1.526 & 3.098 & 2.290 & 2.444 & 2.596 & 2.449 & 2.021 \\
& \textbf{OpenAI} & 0.842 & 1.136 & 1.515 & 2.235 & 1.643 & 1.638 & 1.598 & 1.593 & 1.525 \\
\midrule
\multirow{3}{*}{\textbf{LDPM}} & \textbf{LDA} & 0.310 & 0.314 & 0.404 & 0.568 & 0.541 & 0.542 & 0.565 & 0.587 & 0.479 \\
 & \textbf{BERT} & 0.301 & 0.303 & 0.431 & 0.570 & 0.544 & 0.541 & 0.560 & 0.593 & 0.480 \\
 & \textbf{OpenAI} & 0.315 & 0.341 & 0.356 & 0.498 & 0.578 & 0.522 & 0.518 & 0.511 & 0.455 \\
\bottomrule
\end{tabular}
\end{table}

The traditional linear panel model that relies solely on macroeconomic covariates exhibits relatively weak performance across all forecast horizons, with an average PMSE of 2.365, which shows the incapability of low-dimensional parametric macroeconomic panel modeling to capture the complex dynamics of regional inflation. 
Augmenting the linear panel model with textual embeddings leads to improvements in forecasting accuracy, although the performances vary across embedding strategies. The LDA-based embeddings deliver relatively stable gains across horizons, while higher-dimensional embeddings such as BERT and OpenAI tend to introduce additional noise, resulting in cumulating errors especially for longer horizons. The limitation of such modeling procedures results from the inclusion of irrelevant posts and parametric modeling restriction. 

The proposed LDPM uniformly outperforms the existing methods. The substantial gains are primarily driven by the inclusion of the involvement of LLM-powered surrogate and the new DPT learning strategy. The former increments current low-frequency outcomes with limited information with the generated high-frequency surrogate, fully utilizing the unobserved wisdom of LLM without introducing irrelevant noise. The latter enables flexible non-linear learning of high-dimensional textual embeddings, allowing the model to distill economically meaningful semantic signals that are strongly associated with inflation dynamics. In general, the proposed LDPMs consistently outperform their linear counterparts across all embedding types. In particular, the LDPM based on OpenAI embeddings performs exceptionally well, achieving uniformly lowest PMSE across all horizons. This also highlights the strength of general-purpose OpenAI embeddings in capturing broad and forward-looking economic narratives when coupled with an appropriate deep learning architecture. Furthermore, although forecast errors generally increase with the horizon length for all models, LDPM displays markedly greater stability. The LDPM with OpenAI exhibits the flattest error profile, maintaining low PMSE values even at longer horizons. 

\subsection{Conformal Prediction Intervals}

To further assess the uncertainty associated with regional inflation forecasts, we construct prediction intervals using the aforementioned conformal inference. In line with our earlier findings that OpenAI embeddings consistently deliver the strongest predictive performance, we restrict attention here to OpenAI embedding based models, and omit the LDA and BERT embeddings for clearer illustration.

To conduct conformal inference, the full sample is now divided into four temporally ordered subsets:
 the training set used for model estimation, the validation set for hyperparameter tuning, the calibration set used exclusively for constructing conformal prediction intervals, and the testing set for out-of-sample evaluation.
The calibration set consists of the 12 months immediately preceding the testing period, while the testing window spans an $H$-month horizon ending in December 2023. This strictly sequential splitting scheme ensures that prediction intervals are constructed without any look-ahead bias.

{\scriptsize
$$
\underbrace{\text{Training set}}_{\text{From Jan 2019}} 
\longleftrightarrow
\underbrace{\text{Validation set for tuning}}_{\text{6 months prior to calibration period}} 
\longrightarrow
\underbrace{\text{Calibration set for conformal}}_{\text{12 months prior to testing period}} 
\longrightarrow
\underbrace{\text{Testing set for prediction}}_{\text{$H$-month period ending in Dec 2023}}
$$
}

Using the calibration residuals, we apply conformal prediction to obtain finite-sample-valid prediction intervals for regional CPI growth. Figure~\ref{fig:prediction_interval}  illustrates the resulting conformal prediction intervals across provinces in December 2023 in the case of $H$=4. 

\begin{figure}[!ht]
    \centering
    \includegraphics[width=1\linewidth]{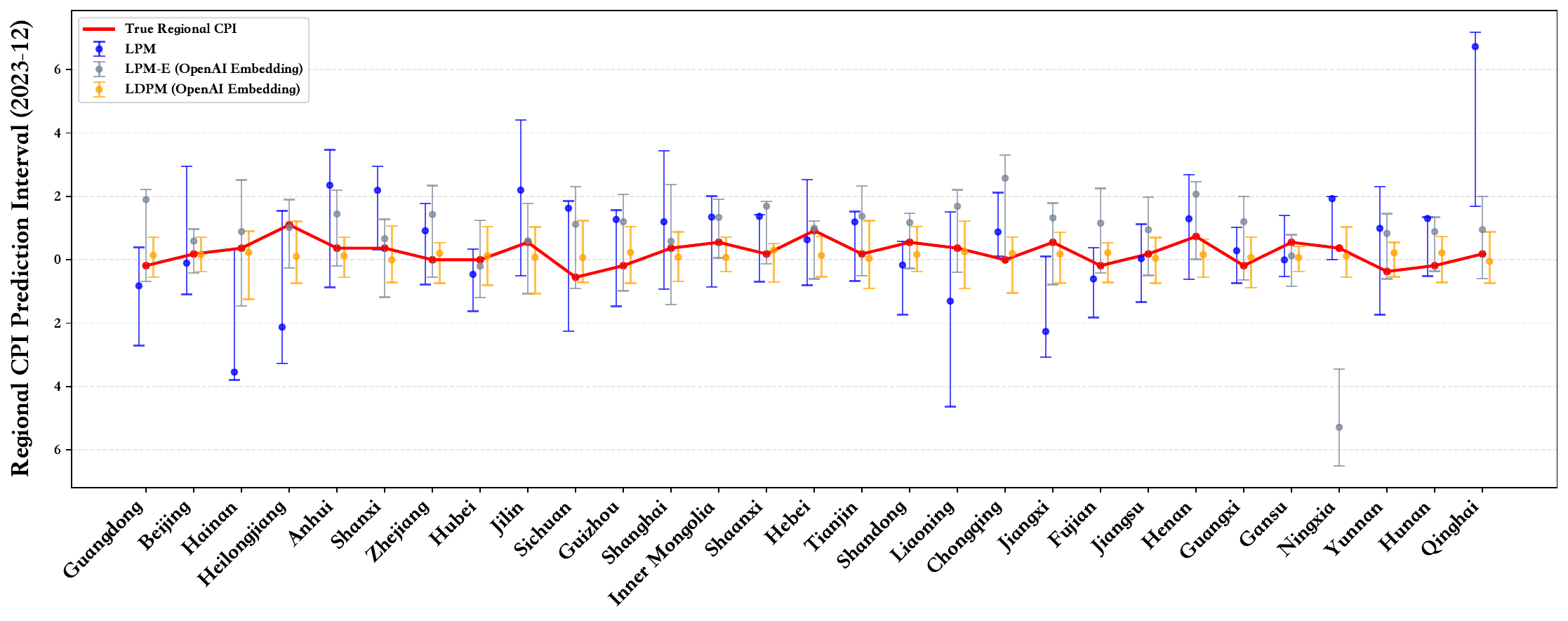}
    \caption{Conformal prediction intervals for regional CPI forecasts based on OpenAI-embedding models. }
    \label{fig:prediction_interval}
\end{figure}

We observe that the traditional linear panel model LPM exhibits wide prediction intervals across most regions, reflecting substantial predictive uncertainty when relying solely on macroeconomic information. While incorporating OpenAI embeddings into the linear panel framework (LPM-E) narrows the intervals in some regions, the resulting prediction intervals remain highly variable and occasionally excessively wide, indicating that high-dimensional embeddings may introduce additional noise when constrained by a linear structure. 

In contrast, LDPM consistently delivers the most compact and stable prediction intervals across provinces. The intervals generated by this model are visibly tighter and more evenly distributed around the realized CPI values, suggesting a substantial reduction in predictive uncertainty. Importantly, this improvement is achieved without sacrificing coverage, highlighting the effectiveness of the Deep Panel architecture in transforming rich AI-based embeddings into economically meaningful predictive signals. Moreover, LDPM demonstrates greater robustness to regional heterogeneity. Even in provinces where the linear models exhibit pronounced uncertainty or extreme interval widths, the LDPM specification maintains well-calibrated and relatively narrow prediction intervals. In sum, the resulting conformal prediction intervals validate the robustness and reliability of the point forecasting of regional CPI provided by LDPM.

\subsection{Sources of Predictive Power}
For deeper understanding how {\it Weibo} narratives provide insights for the CPI forecasting with the aid of LLM, we analyze the topic keywords from the LDA embedding. 
Figure \ref{fig:topic_words_cloud} depicts the top 30 topic words of the 20 LDA topics from the filtered inflation-related {\it Weibo} posts. For interpretability, we assign economically meaningful labels to each topic based on their most representative keywords and relative weights in the topic model. It shows how the discussions on {\it Weibo} cluster into distinct topics that reflect economic forces relevant to CPI forecasting. For instance, persistent chatter about ``price increases'' denotes rising inflation expectations in the public, often foreshadowing official price index upticks, whereas talk of ``decline'' alongside weak demand indicators suggests emerging deflationary pressure. Multiple housing-related topics, such as ``mortgage rate'' and some crisis terms like ``Evergrande'', reflect real estate market stress that dampens consumer spending and suppresses prices in related sectors. Another topic captures wages and living costs (e.g. ``wage'' and ``price level''), revealing consumer sentiment on purchasing power, which can either curb consumption or, if confidence shifts, contribute to a wage-price spiral. {\it Weibo} discourse also anticipates policy shifts, keywords such as ``interest rate cuts'' or local housing measures (``home recognition'', ``loan recognition'' and ``purchase restriction'') signal expected monetary easing and fiscal support aimed at stimulating demand, while attention to global cues (``Federal Reserve rate hikes'', ``global inflation'') ties external cost-push factors to domestic price expectations. The economic narrative theory states that public economic sentiment can drive real outcomes\citep{shiller2020narrative}. Therefore, surges or lulls in these topics serve as early warnings of inflationary or deflationary trends, thereby enrich CPI forecasts with timely insight into evolving demand-supply conditions and inflation expectations.

\begin{figure}[!ht]
    \centering
    \includegraphics[width=1\linewidth]{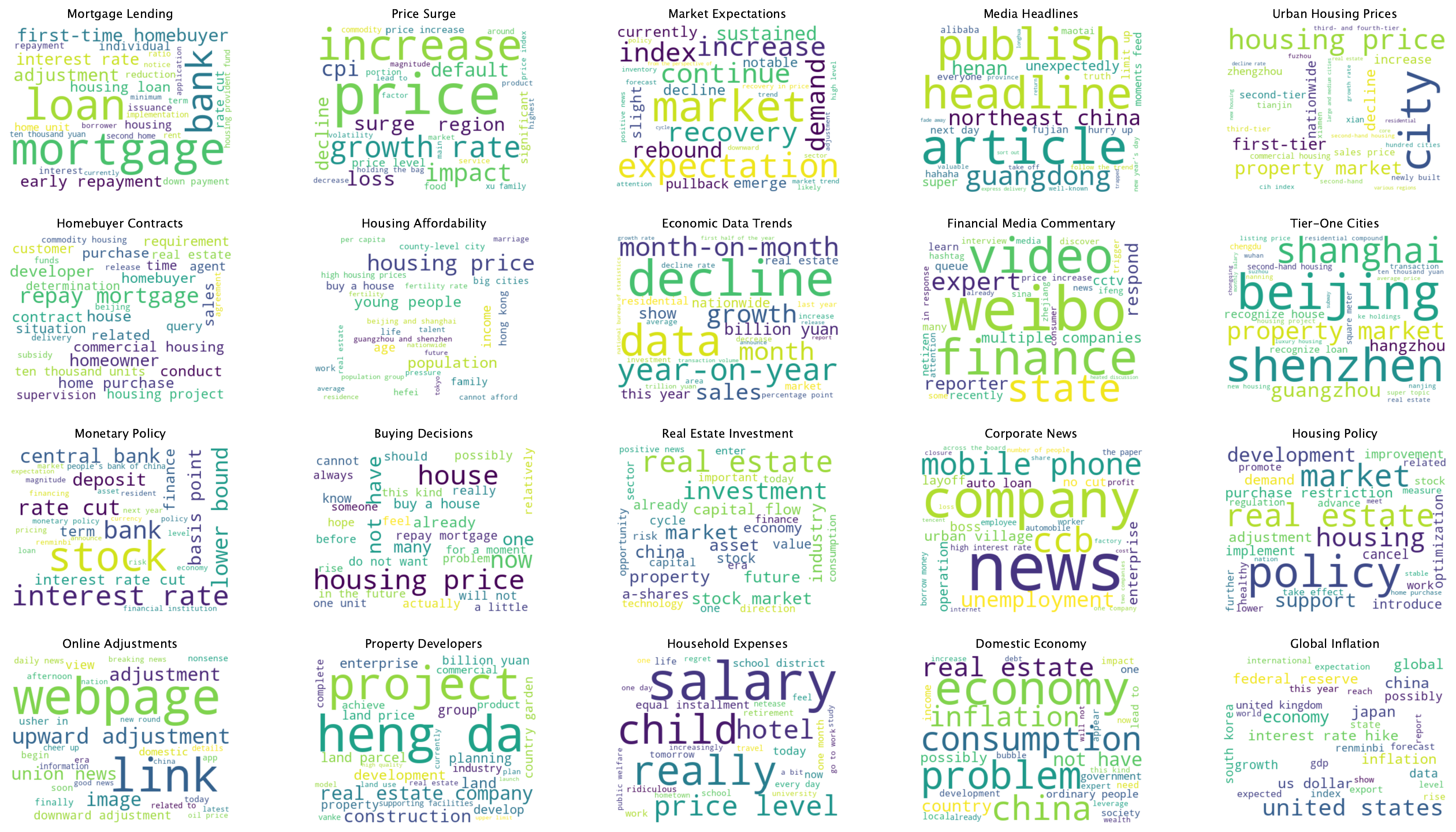}
    \caption{Topic word clouds. The 20 LDA topics are estimated from the full 2019-2023 {\it Weibo} corpus. Each panel presents a word cloud for one topic. The size of each word is proportional to its estimated importance within the topic, measured by the topic-word weight derived from the LDA model. Larger words indicate higher weights and greater relevance to the corresponding topic, while colors are used solely for visual distinction and carry no quantitative meaning.
}
    \label{fig:topic_words_cloud}
\end{figure}

We then examine the temporal distributions of the 20 LDA topics estimated from the full 2019-2023{
\it Weibo} corpus, which allows us to track when particular narratives become more salient in public discussion and how these surges coincide with turning points in inflation. 
Figure \ref{fig:topic_time_pattern} shows that {\it Weibo} narratives about inflation are highly time-varying, and this volatility pattern helps explain why they improve CPI forecasting. Some topics, such as ``Mortgage Lending'', ``Tier-One Cities'' and ``Price Surge'', exhibit pronounced swings, with
sustained elevations around periods of housing policy adjustment, property market stress, while more generic macroeconomy topics (e.g. ``Monetary Policy'',
``Global Inflation'') evolve more smoothly and capture the slow-moving
backdrop. Other topics, such as those centered on distressed developers (`` Property Developers'', ``Buying Decisions''), wage and cost-of-living pressures (``Household Expenses''), display episodic spikes that may coincide with rising household budget stress and changing consumption behavior, which feed directly into components of the CPI. These time-varying patterns mean that sharp increases in specific cost, housing, or income-stress narratives act as high-frequency early warnings of turning points in consumer prices, while the more persistent topics track underlying trend conditions. Because official CPI is low-frequency and backward-looking, the model can exploit the joint dynamics of these volatile and persistent topic weights to detect emerging demand- and supply-side inflation pressures ahead of the official index, thereby improving the timeliness and predictive accuracy of CPI forecasts. To examine region-specific fluctuations in narrative topics, we focus on China's two largest and most representative cities, Beijing and Shanghai, with detailed results reported in the Supplementary Material \ref{sec.dailytopicbjsh}.

\begin{figure}[!ht]
    \centering
    \includegraphics[width=1\linewidth]{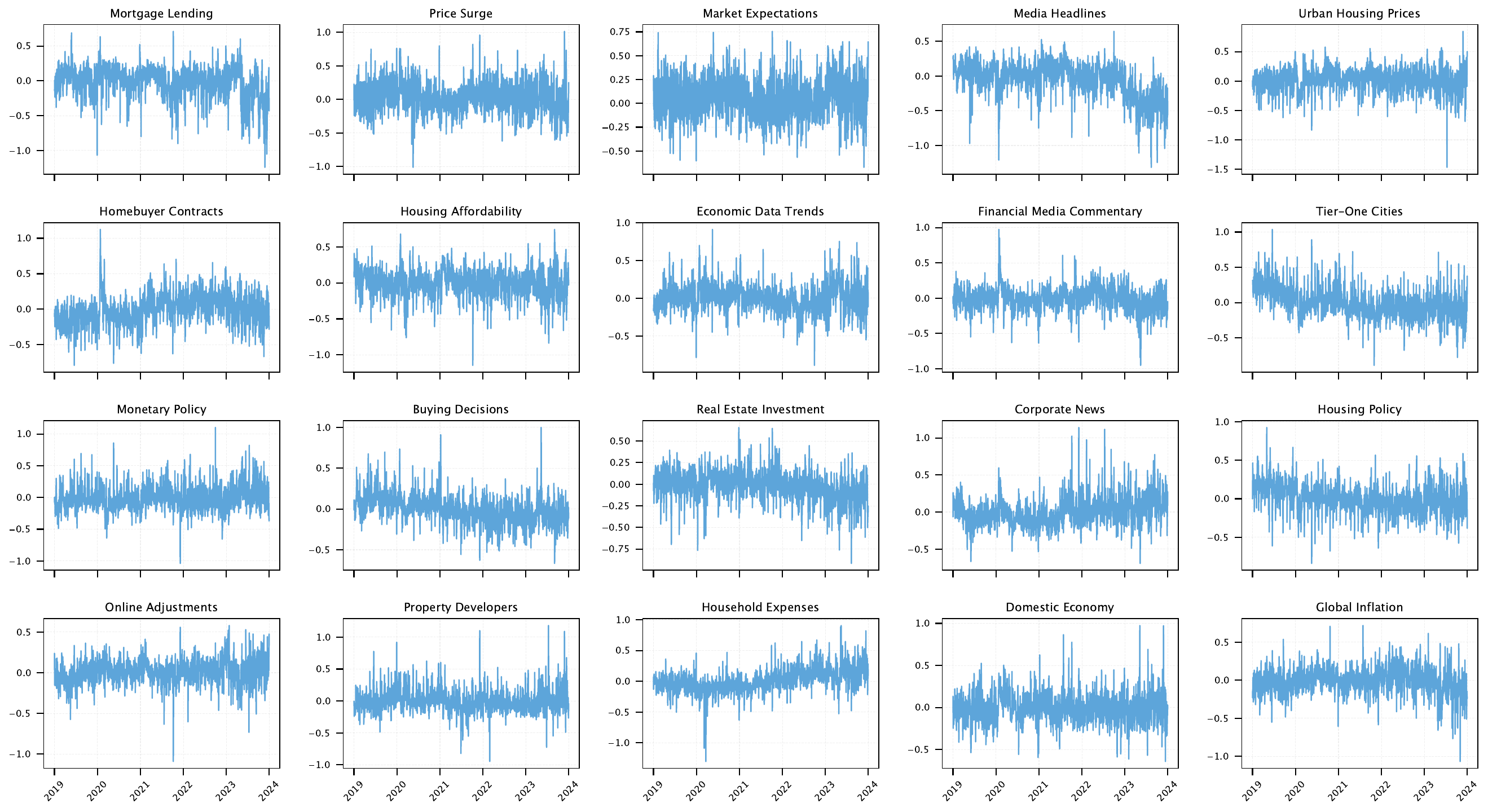}
    \caption{Daily topic probability fluctuations. The 20 LDA topics are estimated from the full 2019-2023 Weibo corpus. Each panel plots the daily average weight of a given topic across all posts, where topic weights are standardized using z-scores. The horizontal axis denotes calendar days, and the vertical axis represents the standardized mean topic weight on each day.
}
    \label{fig:topic_time_pattern}
\end{figure}

So far, we have shown that narratives consistently help predict regional inflation across a wide range of forecasting horizons. A natural next question is where this predictive power comes from. Social media narratives summarize a broad set of economic signals, capturing both how economic agents interpret current conditions and what is happening to prices in the real economy. As a result, narratives may be informative either because they proxy for inflation expectations themselves, because they embed direct information about underlying price movements, or because they combine both channels. In what follows, we disentangle these mechanisms to understand why narratives improve inflation forecasting.

Intuitively, households, entrepreneurs, and bankers face distinct economic constraints and optimization problems. These differences imply heterogeneous information demands in expectation formation. Households primarily make consumption and savings decisions under income constraints and therefore tend to focus on narratives related to consumer prices and living costs. Entrepreneurs face pricing and investment decisions under demand uncertainty and cost volatility, making them more attentive to narratives concerning input prices, sectoral stress, and policy conditions. Bankers operate under credit risk management and financial regulation constraints and thus respond more strongly to narratives about housing-market risk, financial stability, liquidity conditions, and global macro developments. 

To examine how these heterogeneous attention structures translate into expectation formation, we collect survey-based inflation expectations for households, entrepreneurs, and bankers from the quarterly surveys conducted by the People's Bank of China (PBC), the central bank of China.\footnote{Survey data are publicly released by the People's Bank of China. See \url{https://www.pbc.gov.cn/}.} These surveys provide forward-looking measures of inflation expectations for the three groups. Because we cannot observe each individual's full information set, we approximate it using major macroeconomic covariates, CPI, unemployment rate, money supply (M2), residential property sales, and total retail sales, together with social-media-based inflation narratives. Based on model selection (specifically, using LASSO for variable selection), we identify the most informative narrative topics and, following \cite{larsen2021news}, re-estimate the selected variables jointly with macro controls using OLS for statistical inference. Lagged expectations are included as baseline controls. Table \ref{tab:macro_expectations} shows that macroeconomic covariates alone explain a substantial fraction of expectation variation. The adjusted $R^2$ equals 0.613 for consumers, 0.873 for entrepreneurs, and 0.829 for bankers. These results indicate that traditional macro indicators already capture important components of expectation dynamics, particularly for entrepreneurs and bankers.

\begin{table}[!ht]
\footnotesize
\centering
    \caption{OLS regression results of inflation expectations based on macroeconomic covariates}
\label{tab:macro_expectations}
\begin{tabular}{lccc}
\toprule
 & (1) & (2) & (3) \\
 & \makecell{Consumer \\ Expectations} 
 & \makecell{Entrepreneur \\ Expectations} 
 & \makecell{Banker \\ Expectations} \\
\midrule
CPI & $-0.266^{*}$ & $-0.100$ & $-0.112$ \\
Unemployment Rate & $-0.363^{*}$ & $0.538^{***}$ & $0.820^{***}$ \\
M2 Money Supply & $0.312^{***}$ & $-0.129^{***}$ & $-0.185^{***}$ \\
Residential Property Sales & $0.132^{***}$ & $-0.178^{***}$ & $-0.198^{***}$ \\
Total Retail Sales & $-0.235^{***}$ & $0.064$ & $0.136^{**}$ \\
\midrule
Adjusted $R^2$ & $0.613$ & $0.873$ & $0.829$ \\
\bottomrule
\end{tabular}
\end{table}

Table \ref{tab:weibo_narratives} presents the regression results based on inflation narrative topics. In Column~I, where only narrative variables are included, the results show that narratives alone possess considerable explanatory power. The adjusted $R^2$ reaches 0.665 for consumers, 0.787 for entrepreneurs, and 0.660 for bankers. Compared with the macro-only specifications, narrative-only models increase explanatory power by 8.5\% for consumers (0.613 to 0.665), while achieving explanatory levels close to macro-only models for entrepreneurs and bankers. Moreover, specific narrative themes are statistically significant: for consumers, ``Price Surge'' and ``Monetary Policy'' narratives are strongly significant; for entrepreneurs, ``Housing Affordability'' and ``Property Developers'' are significant; and for bankers, ``Housing Affordability'' plays a central role. These findings suggest that narratives encode economically structured information rather than mere sentiment noise. The joint specifications in Column~II provide further insight. When macro controls and narratives enter simultaneously, the adjusted $R^2$ rises from 0.873 to 0.895 for entrepreneurs and from 0.829 to 0.946 for bankers, indicating substantial incremental explanatory power. For bankers in particular, the large improvement (0.829 to 0.946) implies that narrative signals meaningfully augment macro fundamentals in explaining expectation formation. Although some narrative coefficients attenuate once macro variables are introduced, reflecting overlap in information content, several remain significant, especially in the banker regression. This pattern indicates that narratives contain components orthogonal to observable macro aggregates.

\begin{table}[!ht]
\centering
\caption{OLS regression results of inflation expectations based on inflation narrative topics}
\label{tab:weibo_narratives}
\tiny
\begin{tabular}{lcccccc}
\toprule
 & \multicolumn{6}{c}{Inflation Expectations} \\
\cline{2-7}
Topic or Controls 
& \multicolumn{2}{c}{Consumer Expectations}
& \multicolumn{2}{c}{Entrepreneur Expectations}
& \multicolumn{2}{c}{Banker Expectations} \\
 & I & II & I & II & I & II \\
\midrule

Mortgage Lending & 8.116 & -10.556 & $-20.116^{*}$ & 3.999 & -10.675 & $19.623^{*}$ \\
Price Surge & $-32.269^{***}$ & -16.728 &  &  &  &  \\
Market Expectations & -10.652 & -2.285 &  &  &  &  \\
Media Headlines &  &  &  &  &  &  \\
Urban Housing Prices &  &  &  &  &  &  \\
Homebuyer Contracts &  &  &  &  &  &  \\
Housing Affordability &  &  & $-26.106^{**}$ & -10.727 & $-34.165^{**}$ & $-31.586^{**}$ \\
Economic Data Trends &  &  &  &  &  &  \\
Financial Media Commentary & $35.094^{***}$ & 4.539 & -1.367 & 2.454 & -1.402 & -19.233 \\
Tier-One Cities &  &  &  &  &  &  \\
Monetary Policy & $-31.139^{***}$ & -2.415 & 10.931 & -1.851 & 2.893 & 10.830 \\
Buying Decisions & -3.167 & 7.297 & -3.074 & 1.879 & 1.784 & 14.050 \\
Real Estate Investment &  &  & -14.989 & -5.081 &  &  \\
Corporate News & $-72.035^{*}$ & -57.172 &  &  &  &  \\
Housing Policy &  &  & 4.416 & -6.354 &  &  \\
Online Adjustments & -1.352 & 31.397 &  &  &  &  \\
Property Developers &  &  & $-35.028^{**}$ & -0.576 & -29.243 & 8.237 \\
Household Expenses &  &  &  &  & -10.329 & 20.150 \\
Domestic Economy &  &  & 14.700 & -4.181 & 18.897 & -11.130 \\
Global Inflation &  &  &  &  & 0.164 & $17.847^{**}$ \\

\midrule
CPI &  & 0.064 &  & -0.056 &  & -0.226 \\
Unemployment Rate &  & -0.406 &  & $0.398^{*}$ &  & $0.461^{**}$ \\
M2 Money Supply &  & 0.394 &  & -0.161 &  & $-0.548^{**}$ \\
Residential Property Sales &  & 0.177 &  & $-0.204^{*}$ &  & $-0.171^{**}$ \\
Total Retail Sales &  & -0.222 &  & 0.063 &  & $0.200^{**}$ \\

\midrule
Macroeconomic Controls & No & Yes & No & Yes & No & Yes \\
Number of Topics & 9 & 9 & 10 & 10 & 10 & 10 \\
Adjusted $R^2$ & 0.665 & 0.610 & 0.787 & 0.895 & 0.660 & 0.946 \\
Screening Method & LASSO & LASSO & LASSO & LASSO & LASSO & LASSO \\
\bottomrule
\end{tabular}
\end{table}

Taken together, the evidence supports two complementary mechanisms. First, narratives aggregate dispersed, high-frequency microeconomic signals, such as housing-market stress, sectoral cost pressures, and financial risk, that are not yet fully incorporated into official macro statistics. Second, narratives encode interpretative frames that shape belief updating. Expectation formation depends not only on observing macro variables but also on interpreting the persistence and nature of economic shocks. Social media narratives reflect and propagate these interpretations. These suggest that inflation narratives constitute an independent information channel that both overlaps with and extends beyond traditional macroeconomic indicators. The heterogeneous significance patterns across households, entrepreneurs, and bankers further indicate that narrative information operates through group-specific attention structures rather than through mechanical correlation with aggregate macro variables.

\section{Simulation Studies} \label{Sec.simulation}

In this section, we conducts simulation studies to evaluate the performance of the proposed LDPM. The simulations are based on synthetic data constructed to mimic the data structure of our real data analysis. 
We compare our LDPM method with several alternative specifications, including LPM and LPM-E mentioned in Section \ref{oos}, and the performance of prediction is evaluated by the PMSE. Given a prediction horizon of $H$ months, we train each model using observations from January 2019 through December 2023 minus $(H+1)$ months. The out-of-sample predictive performance is then evaluated on a testing sample covering the final $H$ months, spanning from December 2023 minus $(H-1)$ months through December 2023. The experiments assess model performance under varying degrees of correlation between the error terms of the target and surrogate models. 

The data generating process is as follows. The outcome $y_{i,t}$ and its surrogate $y_{i,t,k}^S$ are generated based on the nonlinear panel models:

\begin{equation}\label{equ:sim_nlinear1}
    y_{i,t} = \beta_i^{\top}Z(x_{i,t}) + \epsilon_{i,t} 
\end{equation}
\begin{equation}\label{equ:sim_nlinear2}
    y_{i,t,k}^{S} = \theta_{i}^{S,\top}Z(x_{i,t,k}) + \epsilon_{i,t,k}^{S} 
\end{equation}
where $\beta_i^{\top} \in \mathbb{R}^{1 \times p'}$, $\theta_i^{S,\top} \in \mathbb{R}^{1 \times p'}$ and $Z(x) = \cos(W x + b)$, with $W \in \mathbb{R}^{p' \times p}$ denoting the weight matrix and $b \in \mathbb{R}^{p'}$ denoting the phase shift vector, where the cosine function is applied element-wise.
In addition, both $x_{i,t} \in \mathbb{R}^{p}$ and $x_{i,t,k} \in \mathbb{R}^{p}$ are embedding vectors generated by OpenAI's large embedding model with $p=3072$, differing only in their aggregation frequency. We set the coefficients $\beta_i$ and $\theta_i^{S}$ as $p' \times 1$ vectors of slope parameters, assume a latent group structure, as shown in Equ~\eqref{beta-structure-nonlinear} and \eqref{theta-structure-nonlinear}.

\begin{equation}\label{beta-structure-nonlinear}
    \beta_i = \begin{cases}
        w_1 & \text{if i } \in \mathcal{G}_1 \\
        \vdots & \vdots \\
        w_{K_0} & \text{if i } \in \mathcal{G}_{K_0}
    \end{cases}
\end{equation}
and
\begin{equation}\label{theta-structure-nonlinear}
    \theta_{i}^{S} = \begin{cases}
        w^{S}_1 & \text{if i } \in \mathcal{G}_1 \\
        \vdots & \vdots \\
        w^{S}_{K_0} & \text{if i } \in \mathcal{G}_{K_0}
    \end{cases}
\end{equation}

The errors for the target and surrogate models are also generated from a multivariate normal distribution $(\epsilon_{i,t}, \epsilon_{i,t,1}^{S} ,\dots, \epsilon_{i,t,K}^{S})^{\top} \sim \mathcal{N}_{K+1}(0, \boldsymbol{\Sigma})$, where covariance matrix $\boldsymbol{\Sigma}$ of errors consists of equal value $\rho$ except the diagonal elements are all 1. The coefficients $\beta_i$ and $\theta_i^{S}$ are treated as fixed rather than random.
We consider three levels of correlation, $\rho = 0.2, 0.5$ and $0.8$, and repeat each simulation 500 times.

Note that it can be cumbersome to include all embedding features generated by LLMs or BERT in the linear prediction models. Unlike traditional covariates, embedding features are abstract representations in a high-dimensional semantic space, making variable screening or feature selection methods conceptually meaningless. Therefore, for linear-based models, we apply a dimension reduction procedure in Section \ref{sec.embeddingDR} of the supplementary material to obtain a compact representation of these embeddings. It is worth noting that our proposed LDPM is inherently capable of handling high-dimensional embeddings through adaptive representation learning, so this dimensionality reduction step is unnecessary in that context.

\begin{table}[!htbp]
\caption{The average PMSE of various nonlinear forecasting methods across different prediction horizons $H$ and correlation levels $\rho$ under nonlinear case}
\label{res:simulation:nonlinear}
\centering
\small
\begin{tabular}{lccccccccc}
\toprule
\textbf{Method} & \textbf{$8$} & \textbf{$9$} & \textbf{$10$} & \textbf{$11$} & \textbf{$12$} & \textbf{$13$} & \textbf{$14$} & \textbf{$15$} & \textbf{Ave.}\\ 
\midrule
\multicolumn{10}{l}{\textbf{Setting} $\rho=0.2$} \\
\midrule
LPM  & 1.832 & 1.732 & 3.242 & 3.412 & 3.282 & 2.433 & 3.340 & 3.240 & 2.814 \\
LPM-E & 1.503 & 1.510 & 1.526 & 1.670 & 1.805 & 1.943 & 1.791 & 1.754 & 1.688 \\
LDPM & 0.978 & 1.025 & 1.040 & 1.036 & 1.040 & 1.014 & 1.019 & 1.329 & 1.060 \\
\midrule
\multicolumn{10}{l}{\textbf{Setting} $\rho=0.5$} \\
\midrule
LPM & 1.668 & 1.591 & 3.294 & 3.308 & 3.231 & 2.219 & 3.199 & 3.101 & 2.701 \\
LPM-E & 1.564 & 1.493 & 1.494 & 1.616 & 1.874 & 1.946 & 1.844 & 1.763 & 1.699 \\
LDPM & 0.932 & 0.925 & 0.949 & 0.946 & 0.951 & 0.951 & 0.960 & 0.971 & 0.948 \\
\midrule
\multicolumn{10}{l}{\textbf{Setting} $\rho=0.8$} \\
\midrule
LPM & 1.565 & 1.634 & 3.198 & 3.339 & 3.217 & 2.218 & 3.204 & 3.018 & 2.674 \\
LPM-E & 1.625 & 1.568 & 1.585 & 1.528 & 1.704 & 1.923 & 1.907 & 1.817 & 1.707 \\
LDPM & 0.831 & 0.866 & 0.936 & 0.874 & 0.846 & 0.890 & 0.893 & 0.889 & 0.878 \\
\bottomrule
\end{tabular}
\end{table}

Table~\ref{res:simulation:nonlinear} reports the average PMSE for all forecasting methods across different prediction horizons $H$ and correlation levels $\rho$. The proposed LDPM method consistently delivers the lowest PMSE across all horizons and all correlation settings, substantially outperforming both LPM and LPM-E. This demonstrates the superior ability of LDPM to capture complex nonlinear relationships between high-dimensional textual embeddings and regional inflation dynamics. The performance advantage of LDPM becomes increasingly pronounced as the correlation parameter $\rho$ rises from 0.2 to 0.8. Since $\rho$ controls the dependence between the target and surrogate error terms, higher values correspond to stronger shared latent signals between macroeconomic outcomes and LLM-based narrative embeddings. The steady decline in LDPM's average PMSE from 1.060 ($\rho=0.2$) to 0.948 ($\rho=0.5$) and further to 0.878 ($\rho=0.8$) indicates that LDPM is particularly effective at exploiting strengthened semantic signals embedded in the auxiliary LLM-based predictors. 

Furthermore, LDPM maintains stable and robust performance across different forecast horizons, whereas LPM suffers from severe performance degradation at medium and long horizons, reflecting the inability of purely macroeconomic linear specifications to capture persistent nonlinear dynamics. LPM-E provides moderate improvements over LPM by incorporating narrative information, but its gains are clearly dominated by LDPM, which fully integrates LLM-based embeddings through a nonlinear joint modeling framework.

Overall, these results provide strong empirical evidence that LDPM not only achieves uniformly superior predictive accuracy, but also becomes increasingly advantageous as the strength of the latent LLM signal grows. This confirms that the proposed framework effectively transforms high-dimensional semantic information into economically meaningful predictive signals, particularly in environments where textual narratives carry stronger forecasting content.

\section{Conclusion and Discussion} \label{sec.discu}

This study proposes a statistically structured framework, called {\it LLM-powered Deep Panel Modeling (LDPM)} that integrates LLM-powered economic narratives with deep panel learning to forecast regional inflation. By extracting high-frequency, region-specific surrogate indicators from social media text using LLM and a series of fine-tuned BERT models, and jointly modeling them with official CPI data, the proposed approach achieves substantial improvements in out-of-sample predictive performance. The gains are particularly pronounced at long forecasting horizons.

Methodologically， the proposed LDPM introduces a unified framework for (1) utilizing the (latent) wisdom of LLM to generate surrogate outcomes that contain valuable information for the primary outcome of interest, (2) connecting the target and surrogate model through a new surrogate augmentation strategy, (3)  developing a modified DNN training procedure for panel data, called {\it Deep Panel Training (DPT)}, that realizes homogeneity pursuit upon deep panel learning and (4) providing conformal prediction inference to obtain the prediction intervals for the improved point forecasts.

The empirical results highlight the importance of both surrogate augmentation from social media via LLM and the new DPT algorithm. First, social media narratives appear to aggregate diverse economic signals, including consumer price experiences, expectations, policy interpretations, and localized market conditions. Compared with macroeconomic covariates, these narratives are available at much higher frequency and react more rapidly to evolving economic conditions, enabling earlier detection of inflationary pressures.
Second, compared with linear panel models or region-by-region estimation, the DPT more effectively captures dynamic patterns and spatial correlations across regions. Importantly, the performance gains from narrative-based surrogate variables are not uniform across regions, suggesting heterogeneous sensitivity of regional inflation dynamics to narrative information. Such heterogeneity likely reflects differences in economic structure, industrial composition, consumption patterns, and information diffusion across regions.

From a practical perspective, the results suggest that LLM-based narrative indicators can serve as a valuable complement to official inflation statistics, particularly for real-time monitoring and short-term forecasting. Early identification of region-specific inflation pressures may enhance the timeliness of policy responses related to monetary policy implementation, price stabilization, and social welfare adjustments. Last but not least, the proposed framework is also highly extensible. Beyond social media data, it can incorporate alternative textual sources such as news articles, policy documents, or online forums. Moreover, the methodology can be readily adapted to other regional macroeconomic outcomes, including employment, consumption, or housing price dynamics.

\bibliographystyle{apalike}
\bibliography{ref}
\newpage
\appendix
\begin{center}{\bf \Large Supplementary Material to ``LLM-Powered Deep Panel Modeling with Application to Regional CPI Prediction''}

\bigskip
\end{center}

This Supplementary Material contains additional details on the real data analysis, the identification and optimization of the LDPM, the theoretical guarantees for the LDPM conformal prediction intervals, and a discussion of shortcut learning. For a matrix $\bW \in \mathbb{R}^{p \times p}$, we denote by $\bW(:,j) \in \mathbb{R}^p$ its $j$th column and by $\bW(j,:) \in \mathbb{R}^p$ its $j$th row. For a set $\mathcal{A}$, $|\mathcal{A}|$ denotes the cardinality of $\mathcal{A}$.

\section{Additional details of real data analysis}
\subsection{Data and High-frequency Panel Inflation Index Construction}
\label{app:data}

During the raw data acquisition stage, we collect publicly available posts from {\it Sina Weibo} using keyword-based retrieval and time-stamped queries designed to capture discussions related to prices, consumption, and cost-of-living conditions. To ensure full compliance with ethical standards and data protection principles, all personally identifiable information—including usernames, user IDs, profile descriptions, and any auxiliary user metadata—is permanently removed at ingestion. The retained fields are limited to the post text, posting timestamp, and geographic identifier (province or city), yielding a fully anonymized corpus. We then apply a sequence of preprocessing filters to eliminate invalid records, including duplicated texts, automated spam content, posts with corrupted encoding, and entries lacking valid timestamps. 

Table~\ref{tab:weibo_summary} reports detailed summary statistics of the {\it Sina Weibo} dataset for the 2019--2023 period. Each observation in the raw dataset corresponds to an individual Weibo post and contains a unique post identifier, a time stamp at the daily frequency, textual content, and user interaction measures, including reposts, comments, likes, topic tags, and @-mentions. Over this five-year period, the raw corpus comprises approximately 119.8 million posts, reflecting the large scale and high engagement of public discourse on the platform.

Starting from the raw corpus, we apply a sequence of preprocessing steps to remove invalid records, duplicated posts, and entries with missing time or location information, yielding approximately 95 million cleaned posts between 2019 and 2023. We then employ an LLM-based advertisement filter (Advertisement-LLM) to identify and exclude commercial and promotional content, which accounts for roughly 24.7 million posts over the sample period. The remaining non-advertisement posts are subsequently classified by a category-level LLM (category-LLM), which assigns each post to a mutually exclusive semantic category.

Across the full 2019--2023 sample, the LLM-based classification identifies approximately 5.8 million posts as inflation-related narratives, which constitute the core textual input for our empirical analysis. Other major categories include lifestyle-related content (about 23.9 million posts), entertainment-related discussions (about 17.4 million posts), sentiment-oriented expressions unrelated to prices (about 9.6 million posts), and news-related content (about 14.1 million posts). These distributions highlight the relative sparsity of inflation-related discussions within the broader social media environment and motivate the use of a multi-stage semantic filtering framework. Together, the cleaned and categorized posts form a high-frequency textual panel that underpins the construction of province-level inflation narrative indicators in the main analysis.

\begin{table}[!ht]
\centering
\tiny
\caption{Annual statistics of {\it Sina Weibo} dataset from 2019 to 2023, including total posts, user interactions, and cleaned data counts.}
\label{tab:weibo_summary}
\begin{tabular}{lrrrrr}
\toprule
\textbf{Category} & \textbf{2019} & \textbf{2020} & \textbf{2021} & \textbf{2022} & \textbf{2023} \\
\midrule
\multicolumn{6}{l}{\textbf{Raw Data Overview}} \\
\midrule
Total Count & 23{,}867{,}397 & 21{,}478{,}002 & 25{,}055{,}762 & 23{,}558{,}065 & 25{,}898{,}679 \\
Reposted Count & 1{,}090{,}153 & 915{,}312 & 1{,}068{,}878 & 1{,}017{,}473 & 1{,}131{,}136 \\
Is-Repost Count & 3{,}429{,}038 & 3{,}119{,}749 & 3{,}525{,}874 & 3{,}215{,}841 & 3{,}733{,}858 \\
Unique Users & 4{,}931{,}671 & 4{,}967{,}343 & 4{,}410{,}642 & 4{,}623{,}454 & 4{,}998{,}838 \\
Unique Topics & 1{,}022{,}506 & 1{,}516{,}800 & 1{,}608{,}734 & 1{,}456{,}434 & 1{,}883{,}884 \\
@ Mentions & 5{,}044{,}893 & 4{,}691{,}122 & 6{,}389{,}873 & 4{,}515{,}820 & 5{,}328{,}074 \\
Total Likes & 12{,}839{,}169{,}800 & 20{,}403{,}910{,}002 & 30{,}258{,}299{,}263 & 20{,}075{,}445{,}906 & 21{,}861{,}713{,}953 \\
Total Reposts & 8{,}648{,}416{,}634 & 9{,}356{,}530{,}796 & 19{,}463{,}395{,}462 & 13{,}335{,}659{,}391 & 14{,}777{,}487{,}978 \\
Total Comments & 3{,}580{,}772{,}240 & 3{,}616{,}734{,}933 & 6{,}134{,}589{,}573 & 4{,}715{,}755{,}638 & 4{,}895{,}875{,}472 \\
\midrule
\multicolumn{6}{l}{\textbf{Cleaned Data Overview}} \\
\midrule
Total (Raw) & 23{,}867{,}397 & 21{,}478{,}002 & 25{,}055{,}762 & 23{,}558{,}065 & 25{,}898{,}679 \\
After Cleaning & 17{,}083{,}975 & 17{,}795{,}385 & 18{,}443{,}240 & 19{,}909{,}625 & 22{,}218{,}395 \\
Advertisement (LLM1) & 4{,}761{,}493 & 4{,}327{,}198 & 4{,}906{,}332 & 5{,}068{,}002 & 5{,}643{,}890 \\
Inflation (LLM2-0) & 1{,}132{,}086 & 1{,}082{,}814 & 1{,}197{,}545 & 1{,}263{,}611 & 1{,}114{,}401 \\
Lifestyle (LLM2-4) & 4{,}363{,}671 & 4{,}272{,}760 & 4{,}306{,}088 & 5{,}408{,}548 & 5{,}533{,}003 \\
Entertainment (LLM2-1) & 2{,}810{,}251 & 2{,}981{,}294 & 3{,}359{,}477 & 3{,}200{,}259 & 5{,}069{,}762 \\
Sentiment (LLM2-2) & 1{,}723{,}666 & 1{,}928{,}517 & 1{,}882{,}868 & 1{,}960{,}937 & 2{,}100{,}086 \\
News (LLM2-3) & 2{,}292{,}808 & 3{,}202{,}802 & 2{,}790{,}930 & 3{,}008{,}268 & 2{,}757{,}253 \\
\bottomrule
\end{tabular}
\end{table}

Based on the filtering inflation-related results of the three large models mentioned in the main text, we further extract and standardize all city- and location-level references appearing in the post texts and map them to province-level administrative units using a curated geographic dictionary. If a post mentions multiple cities or provinces, we interpret it as implicitly discussing inflation-related conditions in multiple regions and assign the post to each corresponding province. This procedure yields a subset of slightly over two million posts with explicit province-level geographic information, which forms the initial analytical corpus for downstream tasks. Aggregating these posts by province and time, we construct a high-frequency province–time panel dataset that underpins the subsequent narrative extraction, sentiment quantification, and predictive modeling exercises.

\subsection{Embedding Dimension Reduction}\label{sec.embeddingDR}

Let $\mathbf{X}$ denote the $NT \times p$ matrix containing embedding vectors, where each row corresponds to one observation for a specific region and time period, and $p$ is the embedding dimension (e.g., $p=3072$ for the OpenAI embeddings and $p=768$ for the BERT embeddings). For any $1 \le r \le \min\{NT, p\}$, let $\sigma_r$ be the $r$-th largest singular value of $\mathbf{X}$, and let $u_r \in \mathbb{R}^{NT}$ and $v_r \in \mathbb{R}^{p}$ denote the corresponding left and right singular vectors, respectively. Fixing an integer $1 \le r_0 < \min\{NT, p\}$, we obtain a rank-$r_0$ singular value decomposition (SVD) approximation of $\mathbf{X}$ as
$$
\mathbf U  \mathrm{diag}(\sigma_1, \ldots, \sigma_{r_0}) \mathbf{V^\top} 
= \sum_{i=1}^{r_0} \sigma_i \mathbf{u}_i \mathbf{v}_i',
$$
where $\mathbf{U} \in \mathbb{R}^{NT \times r_0}$ and $\mathbf{V} \in \mathbb{R}^{p \times r_0}$. 
We then construct the post-SVD feature matrix as $\mathbf{X}\mathbf V$, which provides an $NT \times r_0$ low-dimensional representation of the original embeddings for subsequent regression analysis.
For the LDA embedding features, since each dimension corresponds to a distinct semantic topic, conventional screening or selection procedures can still be meaningfully applied to identify the most relevant topics for inflation prediction.

\subsection{Daily Topic Probability Fluctuations Across Regions}\label{sec.dailytopicbjsh}

Figure \ref{fig:topic_time_pattern_bj} and \ref{fig:topic_time_pattern_sh} report how these topic shares evolve over time and across China's two largest and most representative cities—Beijing and Shanghai (results for other cities are reported in the supplementary materials). These plots show that inflation-relevant narratives on {\it Weibo} are not only noisy but also regionally structured and economically interpretable. In Beijing, several topics display episodic spikes and sharp high-frequency bursts, particularly those associated with policy signals, housing finance, and macro news (visible as short-lived upward jumps in \textit{Housing Affordability}, \textit{Tier-One Cities}, \textit{Corporate News}, and \textit{Domestic Economy}), reflecting Beijing's role as a policy and financial information hub. These bursts tend to coincide with periods of monetary easing, property-market announcements, and international macro shocks, which affect CPI through policy expectations and cost channels rather than direct local prices. By contrast, other Beijing topics (e.g., \textit{Mortgage Lending}, \textit{Market Expectations}, and \textit{Media Headlines}) remain persistent but low-volatility, capturing background macro sentiment rather than immediate shocks, thereby encoding trend information about demand conditions and inflation expectations.

In Shanghai, higher-frequency variability occurs across a broader set of topics (notably \textit{Market Expectations}, \textit{Financial Media Commentary}, \textit{Corporate News}, \textit{Property Developers}, and \textit{Domestic Economy}), with more visible medium-term ramps and clustered volatility, consistent with Shanghai’s status as a consumption-driven and globally integrated market. For example, Shanghai topics linking to commodities, logistics, imported goods, or financial markets exhibit sustained elevations rather than occasional spikes, suggesting that consumer-supply channels and international price pass-through are more continuously reflected in local narratives. Meanwhile, lower-volatility components (e.g., \textit{Price Surge}, \textit{Media Headlines}, and \textit{Homebuyer Contracts}) encode the steady-state consumption environment, such as services and household spending, which feed into CPI's core components.

\begin{figure}[!ht]
    \centering
    \includegraphics[width=1\linewidth]{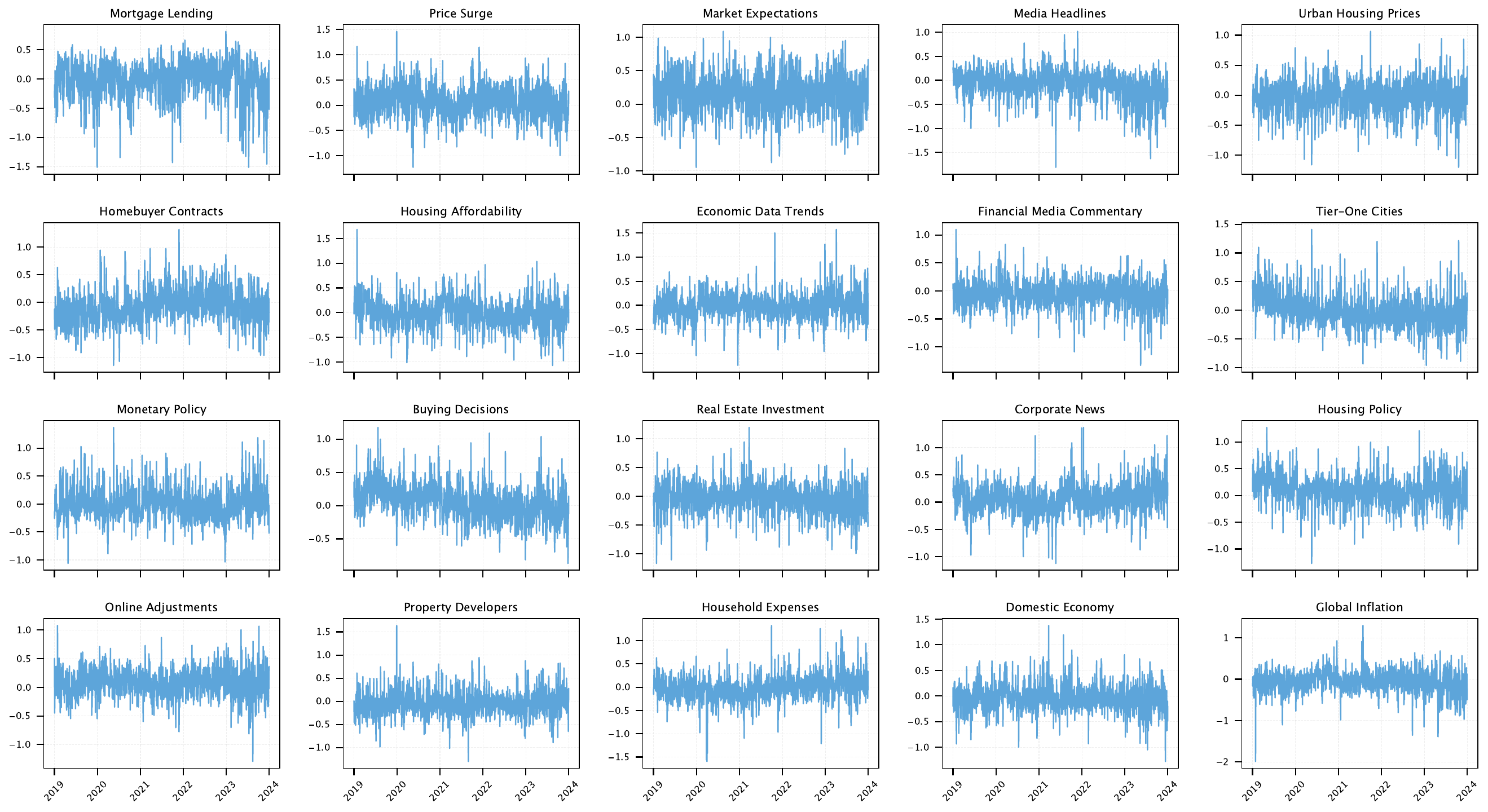}
    \caption{Daily topic probability fluctuations. The 20 LDA topics are estimated from the full 2019---2023 Weibo corpus for Beijing.}
    \label{fig:topic_time_pattern_bj}
\end{figure}

\begin{figure}[!ht]
    \centering
    \includegraphics[width=1\linewidth]{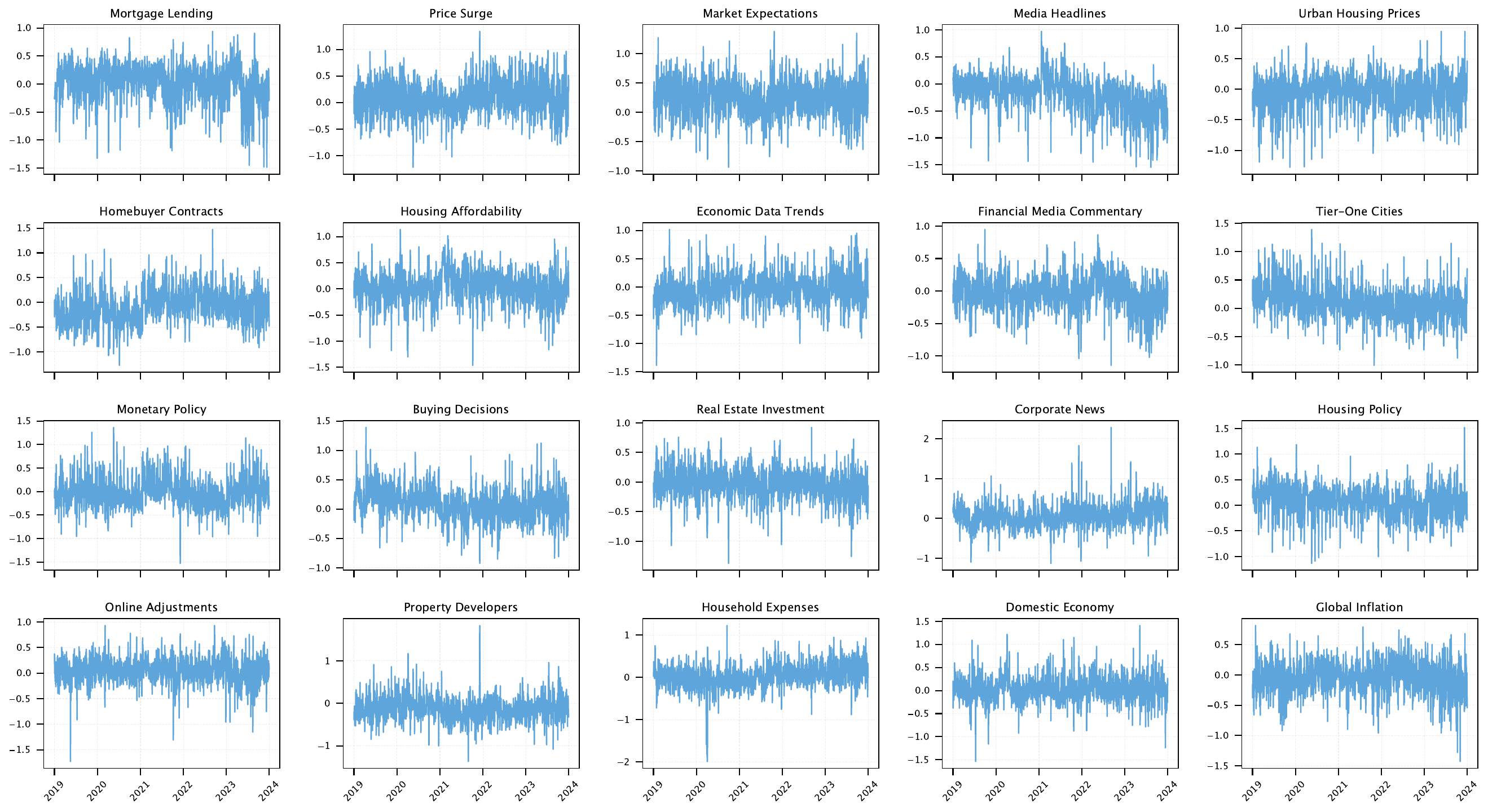}
    \caption{Daily topic probability fluctuations. The 20 LDA topics are estimated from the full 2019---2023 Weibo corpus for Shanghai.}
    \label{fig:topic_time_pattern_sh}
\end{figure}

\section{Identifiability of LLM-powered Deep Panel}
\label{app:identifiability}

We study identifiability for model \eqref{eq:deep-panel-model} in Subsection \ref{sec:deeppanel}
when the shared feature extractor is a ReLU network. As is well known, ReLU networks are not
identifiable at the parameter level because of intrinsic symmetries, most notably permutations
of hidden neurons and positive rescalings. To formalize this non-identifiability, we adopt the
notion of equivalence between parameterizations introduced in Definition~3 of
\cite{bona2023parameter}, which we restate here for completeness.

\begin{definition}[Equivalence between parameter sets]
\label{def:equiv-params}
Consider an $L$-layer feedforward ReLU network
$$
Q^{\mathrm{DNN}}_{\theta}(\bx)
=
\bW^{(L)}\sigma\!\left(
\bW^{(L-1)}\sigma\!\left(
\cdots \sigma\!\left(\bW^{(1)}\bx+\bgamma^{(1)}\right)\cdots
\right)+\bgamma^{(L-1)}
\right)+\bgamma^{(L)},
$$
with parameter vector
$$
\theta=\{(\bW^{(l)},\bgamma^{(l)}):l=1,\ldots,L\}.
$$

Two parameter sets $\theta$ and $\tilde\theta$ are said to be \emph{equivalent}, written
$\theta\sim\tilde\theta$, if they induce the same network function and if $\tilde\theta$ can be
obtained from $\theta$ through a finite composition of the following ReLU symmetries:
\begin{enumerate}
    \item \textbf{Permutation:} for each hidden layer $l=1,\ldots,L-1$, there exists a
    permutation matrix $\bP^{(l)}$ such that
    $$
    \widetilde{\bW}^{(l)}=\bP^{(l)}\bW^{(l)},\qquad
    \widetilde{\bgamma}^{(l)}=\bP^{(l)}\bgamma^{(l)},\qquad
    \widetilde{\bW}^{(l+1)}=\bW^{(l+1)}(\bP^{(l)})^{-1}.
    $$

    \item \textbf{Positive scaling:} for each hidden neuron $j$ in layer $l$, there exists
    $c_{l,j}>0$ such that
    $$
    \widetilde{\bW}^{(l)}(j,:)=c_{l,j}\bW^{(l)}(j,:),\qquad
    \widetilde{\bgamma}^{(l)}_j=c_{l,j}\bgamma^{(l)}_j,
    $$
    and
    $$
    \widetilde{\bW}^{(l+1)}(:,j)=c_{l,j}^{-1}\bW^{(l+1)}(:,j).
    $$
\end{enumerate}
The collection of all such transformations forms the equivalence group $\mathcal G$, and we write
$\theta\sim\tilde\theta$ if and only if there exists $g\in\mathcal G$ such that
$\tilde\theta=g\cdot\theta$.
\end{definition}

By Proposition~4 and Theorem~7 of \cite{bona2023parameter}, under suitable certificate conditions,
equality of two ReLU network functions on a nonempty open set implies equivalence up to the
symmetry group above. In our Deep Panel model,
\begin{equation}
\label{eq:deep-panel-ident}
Q_i^{\mathrm{DNN}}(\bx)
=
\bbeta_i^\top h_\theta(\bx)+b_i,
\qquad i=1,\ldots,N,
\end{equation}
where $h_\theta(\bx)\in\mathbb R^{d_h}$ is a shared ReLU feature map and the unit-specific output
parameters satisfy a latent group structure: there exist distinct group centers
$\{(\boldsymbol\eta_k,\varphi_k)\}_{k=1}^{K_0}$ and a mapping
$g:\{1,\ldots,N\}\to\{1,\ldots,K_0\}$ such that
\begin{equation}
\label{eq:group-structure-ident-short}
\bbeta_i=\boldsymbol\eta_{g(i)},
\qquad
b_i=\varphi_{g(i)},
\qquad i=1,\ldots,N.
\end{equation}
To separate the slope and intercept uniquely from the realized regression function, we further
assume that the shared feature map is affinely nondegenerate on some nonempty open set
$\mathcal U$, namely,
$$
\Span\!\Bigl\{
\begin{pmatrix}
h_\theta(\bx)\\
1
\end{pmatrix}
:\bx\in\mathcal U
\Bigr\}
=
\mathbb R^{d_h+1}.
$$

The next lemma shows that although the shared ReLU representation is identifiable only up to the
usual symmetry group, the induced partition of units remains identifiable up to relabeling.

\begin{lemma}[Identifiability of latent groups]
\label{lem:identifiability-groups}
Assume the certificate conditions in Proposition~4 and Theorem~7 of \cite{bona2023parameter},
the affine nondegeneracy condition above, and the distinctness of the group centers in
\eqref{eq:group-structure-ident-short}. Suppose another parameterization
$(\tilde\theta,\{\tilde\bbeta_i,\tilde b_i\}_{i=1}^N)$ satisfies
\begin{equation}
\label{eq:equal-unit-functions-short}
\tilde\bbeta_i^\top h_{\tilde\theta}(\bx)+\tilde b_i
=
\bbeta_i^\top h_\theta(\bx)+b_i,
\qquad
\forall \bx\in\mathcal U,\quad i=1,\ldots,N,
\end{equation}
and that $\theta\sim\tilde\theta$. Then there exist a permutation matrix $\bP$ and a positive
diagonal matrix $\bD$ such that
\begin{equation}
\label{eq:beta-transform-short}
\tilde\bbeta_i=\bD^{-1}\bP\bbeta_i,
\qquad
\tilde b_i=b_i,
\qquad i=1,\ldots,N.
\end{equation}
Consequently, the latent group assignment $\{g(i)\}_{i=1}^N$ is identifiable up to relabeling of
the group indices.
\end{lemma}

The lemma shows that parameter-level non-identifiability of the shared ReLU feature extractor does
not destroy identifiability of the latent grouping. The reason is that all units inherit the same
permutation-rescaling ambiguity from the shared representation, and this common invertible
transformation preserves equality relations among the group-specific output parameters.

\begin{proof}
Since $\theta\sim\tilde\theta$, the equivalence transformation acts on the final hidden-layer
coordinates by a common permutation and positive rescaling. Hence there exist a permutation matrix
$\bP$ and a positive diagonal matrix $\bD$ such that
\begin{equation}
\label{eq:hidden-map-transform-short}
h_{\tilde\theta}(\bx)=\bD\bP\,h_\theta(\bx),
\qquad \forall \bx\in\mathcal U.
\end{equation}
Fix any $i\in\{1,\ldots,N\}$. Combining \eqref{eq:equal-unit-functions-short} and
\eqref{eq:hidden-map-transform-short}, we obtain
$$
\bbeta_i^\top h_\theta(\bx)+b_i
=
\tilde\bbeta_i^\top \bD\bP\,h_\theta(\bx)+\tilde b_i,
\qquad \forall \bx\in\mathcal U.
$$
Equivalently,
$$
\Bigl(\bbeta_i-\bP^\top\bD\,\tilde\bbeta_i\Bigr)^\top h_\theta(\bx)
+
(b_i-\tilde b_i)
=0,
\qquad \forall \bx\in\mathcal U.
$$
Thus
$$
\begin{pmatrix}
\bbeta_i-\bP^\top\bD\,\tilde\bbeta_i\\
b_i-\tilde b_i
\end{pmatrix}^{\!\top}
\begin{pmatrix}
h_\theta(\bx)\\
1
\end{pmatrix}
=0,
\qquad \forall \bx\in\mathcal U.
$$
By affine nondegeneracy, the vectors
$\bigl(h_\theta(\bx)^\top,1\bigr)^\top$ span $\mathbb R^{d_h+1}$ as $\bx$ varies over
$\mathcal U$, so the only vector orthogonal to all of them is the zero vector. Therefore,
$$
\bbeta_i=\bP^\top\bD\,\tilde\bbeta_i
\qquad\text{and}\qquad
b_i=\tilde b_i.
$$
Multiplying the first identity by $\bP$ and then by $\bD^{-1}$ yields
$$
\tilde\bbeta_i=\bD^{-1}\bP\bbeta_i,
\qquad
\tilde b_i=b_i.
$$
This proves \eqref{eq:beta-transform-short}.

Now define transformed group centers by
$$
\tilde{\boldsymbol\eta}_k:=\bD^{-1}\bP\,\boldsymbol\eta_k,
\qquad
\tilde\varphi_k:=\varphi_k,
\qquad k=1,\ldots,K_0.
$$
Then, by \eqref{eq:group-structure-ident-short},
$$
\tilde\bbeta_i=\tilde{\boldsymbol\eta}_{g(i)},
\qquad
\tilde b_i=\tilde\varphi_{g(i)},
\qquad i=1,\ldots,N.
$$
Because $\bD^{-1}\bP$ is invertible and the original group centers are distinct, the transformed
group centers are also distinct. Hence, for any $i,j$,
$$
g(i)=g(j)
\quad\Longleftrightarrow\quad
(\bbeta_i,b_i)=(\bbeta_j,b_j)
\quad\Longleftrightarrow\quad
(\tilde\bbeta_i,\tilde b_i)=(\tilde\bbeta_j,\tilde b_j).
$$
Therefore the induced partition of $\{1,\ldots,N\}$ is unchanged, except possibly for a relabeling
of the group indices.
\end{proof}

\section{Optimization of the Homogeneity-pursuit Loss}\label{app:optimization}
The objective in \eqref{eq:deep-panel-model} of Subsection \ref{sec:deeppanel} with the C-LASSO-type penalty is nonconvex for two reasons:  
(i) the DNN feature map $h_{i,t}=h(\cdot;\bW,\bgamma)$ is nonlinear in $(\bW,\bgamma)$, and  
(ii) the product-form grouping penalty
$$
\frac{\lambda}{N}
\sum_{i=1}^N
\prod_{k=1}^{K_0}
\left(\big\| (\boldsymbol\eta_k^\top,\varphi_k)^\top - (\bbeta_i^\top,b_i)^\top \big\|_2\right)
$$
is a nonconvex, nonsmooth function of the unit-specific parameters and the group centers.  

We define the full set of trainable parameters as:
\[
\Theta := \Big( \underbrace{\bW, \bgamma}_{\Theta_1}, \underbrace{\{\bbeta_i, b_i\}_{i=1}^N}_{\Theta_2}, \underbrace{\{\boldsymbol\eta_k, \varphi_k\}_{k=1}^{K_0}}_{\Theta_3} \Big).
\]
Given the target variable $y_{i,t}$ (which may represent the original outcome or surrogate residuals from a prior stage), the theoretical target-side loss is:
\[
\mathcal{L}(\Theta)
=
\frac{1}{NT}
\sum_{i=1}^N \sum_{t=1}^T
\bigl(
    y_{i,t} - \bbeta_i^\top h_{i,t}(\Theta_1) - b_i
\bigr)^2
+
\frac{\lambda}{N}
\sum_{i=1}^N
\prod_{k=1}^{K_0}
\big\| (\boldsymbol\eta_k^\top,\varphi_k)^\top - (\bbeta_i^\top,b_i)^\top \big\|_2.
\]

\textbf{Joint End-to-End Optimization.}
While the objective can in principle be optimized via blockwise coordinate descent (alternating between updating the group centers $\Theta_3$ via clustering routines and the remaining parameters via gradient steps), such an approach can be cumbersome within modern deep learning frameworks. Instead, we treat all components of $\Theta$ as trainable parameters and update them jointly end-to-end using mini-batch stochastic gradient descent (SGD) or Adam.

\textbf{Empirical Penalty Substitution.}
Directly minimizing the product-form penalty via backpropagation often leads to numerical instability, such as vanishing or exploding gradients. Therefore, in our implementation, we approximate the product penalty with the minimum Euclidean distance across group centers for each unit. This yields a numerically stable, clustering-like empirical loss. Concretely, for a mini-batch $\mathcal{B} \subset \{(i,t)\}$, we compute:
$$
\mathcal{L}_{\mathcal{B}}(\Theta)
=
\frac{1}{|\mathcal{B}|}
\sum_{(i,t)\in\mathcal{B}}
\bigl(
y_{i,t} - \bbeta_i^\top h_{i,t}(\Theta_1) - b_i
\bigr)^2
+
\frac{\lambda}{N}
\sum_{i=1}^N
\min_{1 \le k \le K_0}
\big\| (\boldsymbol\eta_k^\top,\varphi_k)^\top - (\bbeta_i^\top,b_i)^\top \big\|_2.
$$

At each iteration, we compute the stochastic gradient $\nabla_{\Theta} \mathcal{L}_{\mathcal{B}}(\Theta)$ via automatic differentiation and perform a standard gradient step. By backpropagating through the $\min$ operator, the network effectively identifies the closest current group center for each unit and simultaneously pulls the unit-specific coefficients and the center toward each other. 

Overall, the optimization of the homogeneity-pursuit loss reduces to a standard DNN training procedure combining continuous weight updates with implicit, low-dimensional clustering of the group centers. In our experiments, this joint optimization method converges reliably and scales well to the panel sizes considered in LDPM.

\section{Theoretical Guarantee of Conformal Inference}\label{sec:theoryconformal}
In this section, we discuss the theoretical guarantee of LCPM conformal prediction interval proposed in Subsection \ref{subsec:conformal-PI}, which includes size control and length of prediction interval.

\subsection{Size Control}
In this subsection, we discuss the size of LCPM conformal prediction interval. Throughout this section, we fix a latent group $k$ and a test point
$(i^\ast,T+1)$ with $\widehat g(i^\ast)=k$.
Define the fitted prediction score
$$
s_{i,t}:=\big|y_{i,t}-\widehat y_{i,t}\big|,
\qquad
(i,t)\in\mathcal D_{\mathrm{Cal}}\cup\{(i^\ast,T+1)\},
$$
and let $\widetilde F_k$ denote the empirical CDF of the calibration scores
$\{s_{i,t}:(i,t)\in\mathcal D_{\mathrm{Cal}},\ \widehat g(i)=k\}$.
Let $F_k$ denote the conditional CDF of the test score $s_{i^\ast,T+1}$
given $\widehat g(i^\ast)=k$.

To relate the fitted-score distribution to the underlying data-generating
process, we introduce an oracle score process.
Let $Q_i(\cdot)$ denote the population regression function and define the oracle
residual and score
$$
e_{i,t}:=y_{i,t}-Q_i(\cdot),
\qquad
r_{i,t}:=|e_{i,t}|.
$$
Let $\widetilde G_k$ be the empirical CDF of
$\{r_{i,t}:(i,t)\in\mathcal D_{\mathrm{Cal}},\ \widehat g(i)=k\}$, and let $G_k$
be the conditional CDF of the test oracle score $r_{i^\ast,T+1}$ given
$\widehat g(i^\ast)=k$.

\begin{assumption}[DNN prediction consistency]
\label{ass:dnn-consistency}
There exist sequences $\delta_{T,k}\ge 0$ and
$\rho^{\mathrm{pred}}_{T,k}\in[0,1]$ with $\delta_{T,k}\to 0$ and
$\rho^{\mathrm{pred}}_{T,k}\to 0$, and an event $\mathcal B_{T,k}$ with
$\Pr(\mathcal B_{T,k})\ge 1-\rho^{\mathrm{pred}}_{T,k}$ such that on
$\mathcal B_{T,k}$,
$$
\sup_{\substack{(i,t)\in\mathcal D_{\mathrm{Cal}}\\ \widehat g(i)=k}}
\big|\widehat y_{i,t}-Q_i(\cdot)\big|
\le
\delta_{T,k},
\qquad
\big|\widehat y_{i^\ast,T+1}-Q_{i^\ast}(\cdot)\big|
\le
\delta_{T,k}.
$$
\end{assumption}

Assumption~\ref{ass:dnn-consistency} requires uniform prediction consistency of
the deep learning predictor within each estimated group, both on the calibration
sample and at the test point. Such conditions are standard in conformal
inference with learned predictors and post-selection inference
\citep{xu2023conformal}.

\begin{assumption}[Uniform CDF approximation of residuals]
\label{ass:mixing}
There exist sequences $\varepsilon^{\mathrm{mix}}_{T,k}\ge 0$ and
$\rho^{\mathrm{mix}}_{T,k}\in[0,1]$ with $\varepsilon^{\mathrm{mix}}_{T,k}\to 0$ and
$\rho^{\mathrm{mix}}_{T,k}\to 0$, and an event $\mathcal M_{T,k}$ with
$\Pr(\mathcal M_{T,k})\ge 1-\rho^{\mathrm{mix}}_{T,k}$ such that on
$\mathcal M_{T,k}$,
$$
\sup_{x\in\mathbb R}\big|\widetilde G_k(x)-G_k(x)\big|
\le
\varepsilon^{\mathrm{mix}}_{T,k}.
$$
\end{assumption}
A sufficient condition is that the within-group oracle residuals form a strictly
stationary $\beta$-mixing sequence with summable mixing coefficients, in which
case uniform convergence of empirical CDFs follows from classical
empirical-process results for dependent data
\citep{yu1994rates,bradley2005basic}. Importantly, this assumption is imposed
only on oracle residuals, rather than on observed outcomes or covariates,
allowing for substantial temporal dependence and nonstationarity in the original
data-generating process.

\begin{assumption}[Lipschitz]
\label{ass:lipschitz}
There exists $L_k<\infty$ such that for all $x\in\mathbb R$ and $u\ge 0$,
$$
G_k(x+u)-G_k(x-u)\le 2L_k u.
$$
\end{assumption}
Assumption~\ref{ass:lipschitz} is a mild regularity condition on the oracle
residual distribution, ruling out pathological concentration near a single
point. It holds, for example, when the oracle residual admits a bounded
conditional density, and is commonly used to translate prediction error bounds
into uniform CDF approximation errors. Define
$$
m_k := |\{(i,t) \in \mathcal{D}_{\mathrm{Cal}}: \hat{g}(i)=k\}|.
$$

\begin{theorem}[Conditional coverage]
\label{prop:cdf-uniform-from-primitive}
Under Assumptions~\ref{ass:dnn-consistency}--\ref{ass:lipschitz},
$$
\bigg|
\Pr\left(
y_{i^\ast,T+1}\in \mathcal C_{i^\ast,T+1}(1-\alpha)
\ \middle|\ \widehat g(i^\ast)=k
\right)
-(1-\alpha)
\bigg|
\le
3\varepsilon^{\mathrm{mix}}_{T,k}
+
4L_k\delta_{T,k}
+
\frac{1}{m_k},
$$
with probability at least
$1-\rho^{\mathrm{pred}}_{T,k}-\rho^{\mathrm{mix}}_{T,k}$.
\end{theorem}
Theorem~\ref{prop:cdf-uniform-from-primitive} shows that conditional coverage
errors vanish as long as prediction error and temporal dependence of oracle
scores are sufficiently mild, without requiring exchangeability or stationarity
of the original data-generating process.

\begin{proof}[Proof of Theorem \ref{prop:cdf-uniform-from-primitive}]
\label{app:proof-cdf-uniform}
Let $\mathcal A_{T,k}:=\mathcal B_{T,k}\cap\mathcal M_{T,k}$.
By Assumptions~\ref{ass:dnn-consistency} and~\ref{ass:mixing},
$$
\Pr(\mathcal A_{T,k}) \ge
1-\rho^{\mathrm{pred}}_{T,k}-\rho^{\mathrm{mix}}_{T,k}.
$$
By the triangle inequality, on $\mathcal A_{T,k}$,
\begin{equation}\label{eq:triangle}
\begin{aligned}
\sup_{x\in\mathbb R}|\widetilde F_k(x)-F_k(x)|
\le\ &
\underbrace{\sup_{x}|\widetilde F_k(x)-\widetilde G_k(x)|}_{(I)}
+
\underbrace{\sup_{x}|\widetilde G_k(x)-G_k(x)|}_{\le\varepsilon^{\mathrm{mix}}_{T,k}}
+
\underbrace{\sup_{x}|G_k(x)-F_k(x)|}_{(II)} .
\end{aligned}
\end{equation}
Actually, we can show $(I) \le 2L_k\delta_{T,k}
+2\varepsilon^{\mathrm{mix}}_{T,k}$ in \eqref{eq:term1} and $(II) \le 2L_k\delta_{T,k}$ in \eqref{eq:term2} on $\mathcal A_{T,k}$. This implies
\begin{equation}\label{equ:Jan23:01}
   \sup_{x\in\mathbb R}|\widetilde F_k(x)-F_k(x)|
\le
3\varepsilon^{\mathrm{mix}}_{T,k}
+4L_k\delta_{T,k}.
\end{equation}
Then, applying Lemma \ref{lem:conditional-size}, we obtain the final result. Now, what remains is to bound $(I)$ and $(II)$. We first prove $(II)$; a similar proof argument can also be applied to $(I)$.

\paragraph{Part 1: Bounding $(II)$.}
On $\mathcal B_{T,k}$, write
$$
s_{i,t}
=
|e_{i,t}+u_{i,t}|,
\qquad
u_{i,t}:=Q_i(\cdot)-\widehat y_{i,t},
$$
so that by the reverse triangle inequality,
$$
|s_{i,t}-r_{i,t}|
=
\big||e_{i,t}+u_{i,t}|-|e_{i,t}|\big|
\le
|u_{i,t}|
\le
\delta_{T,k}.
$$
In particular,
$$
|s_{i^\ast,T+1}-r_{i^\ast,T+1}|\le\delta_{T,k}.
$$
Hence, for any $x\in\mathbb R$,
$$
\{r_{i^\ast,T+1}\le x-\delta_{T,k}\}
\subseteq
\{s_{i^\ast,T+1}\le x\}
\subseteq
\{r_{i^\ast,T+1}\le x+\delta_{T,k}\}.
$$
Taking conditional probabilities given $\widehat g(i^\ast)=k$ yields
$$
G_k(x-\delta_{T,k})\le F_k(x)\le G_k(x+\delta_{T,k}),
$$
and therefore
$$
\sup_{x\in\mathbb R}|F_k(x)-G_k(x)|
\le
\sup_x\big\{G_k(x+\delta_{T,k})-G_k(x-\delta_{T,k})\big\}.
$$
By Assumption~\ref{ass:lipschitz},
\begin{equation}\label{eq:term2}
\sup_{x\in\mathbb R}|F_k(x)-G_k(x)|
\le
2L_k\delta_{T,k}.
\end{equation}

\paragraph{Part 2: Bounding $(I)$.}
The same argument applies to calibration scores.
For any $(i,t)\in\mathcal D_{\mathrm{Cal}}$ with $\widehat g(i)=k$,
$$
r_{i,t}\le x-\delta_{T,k}
 \Rightarrow
s_{i,t}\le x,
\qquad
s_{i,t}\le x
 \Rightarrow
r_{i,t}\le x+\delta_{T,k}.
$$
Summing indicators and dividing by $m_k$ gives, for all $x$,
$$
\widetilde G_k(x-\delta_{T,k})
\le
\widetilde F_k(x)
\le
\widetilde G_k(x+\delta_{T,k}).
$$
Consequently,
$$
\sup_{x\in\mathbb R}|\widetilde F_k(x)-\widetilde G_k(x)|
\le
\sup_x\big\{\widetilde G_k(x+\delta_{T,k})-\widetilde G_k(x-\delta_{T,k})\big\}.
$$
Write $\widetilde G_k=G_k+(\widetilde G_k-G_k)$ and apply the triangle inequality:
$$
\widetilde G_k(x+\delta_{T,k})-\widetilde G_k(x-\delta_{T,k})
\le
\big[G_k(x+\delta_{T,k})-G_k(x-\delta_{T,k})\big]
+2\sup_{u\in\mathbb R}|\widetilde G_k(u)-G_k(u)|.
$$
Taking suprema and using Assumptions~\ref{ass:mixing} and~\ref{ass:lipschitz},
\begin{equation}\label{eq:term1}
\sup_{x\in\mathbb R}|\widetilde F_k(x)-\widetilde G_k(x)|
\le
2L_k\delta_{T,k}
+2\varepsilon^{\mathrm{mix}}_{T,k}.
\end{equation}
\end{proof}
\subsection{Conformal Interval Length Comparison}
\label{subsec:power}

We now turn to efficiency considerations and compare the length of conformal
prediction intervals constructed using the proposed joint modeling framework
with those obtained from a target-only model.
Throughout, we fix a group $k$ and a test point $(i^\ast,T+1)$ with
$\widehat g(i^\ast)=k$.

To facilitate comparison of interval lengths, we impose a mild parametric
condition on the oracle residuals.

\begin{assumption}[Gaussian oracle residuals]
\label{ass:gaussian}
Conditional on $\widehat g(i)=k$,
$$
\epsilon_{i,t}\sim\mathcal N(0,\sigma^2_{\epsilon,k}),
\qquad
e_{i,t}\sim\mathcal N(0,\sigma^2_{e,k}).
$$
\end{assumption}
Under the $L_2$ projection $\epsilon_{i,t}=\Gamma(\mathbf\epsilon^S_{i,t})+e_{i,t}$ in
\eqref{eq:proj-gamma}, we have $\sigma^2_{e,k}\le\sigma^2_{\epsilon,k}$ at the
population level. While this assumption can be relaxed to broader classes such
as log-concave distributions \citep{bobkov1999isoperimetric}, we focus on the
Gaussian case for clarity and intuition.

Consider the target-only prediction model~\eqref{eq:model-general}, and let
$\widehat F_i(\cdot)$ denote its fitted predictor. Define the target-only fitted
score
$$
s^{\mathrm{tar}}_{i,t}
:=
\big|y_{i,t}-\widehat F_i(\cdot)\big|,
\qquad
(i,t)\in\mathcal D_{\mathrm{Cal}}\cup\{(i^\ast,T+1)\}.
$$
Let $\widetilde F^{\mathrm{tar}}_k$ be the empirical CDF of
$\{s^{\mathrm{tar}}_{i,t}:(i,t)\in\mathcal D_{\mathrm{Cal}},\widehat g(i)=k\}$, and
define the target-only conformal cutoff
$$
\widehat q^{\mathrm{tar}}_k
:=
\inf\bigl\{x:\widetilde F^{\mathrm{tar}}_k(x)\ge 1-\alpha\bigr\}.
$$
The resulting target-only conformal interval is
$$
\mathcal C^{\mathrm{tar}}_{i^\ast,T+1}(1-\alpha)
=
\bigl[
\widehat F_{i^\ast}(\cdot)-\widehat q^{\mathrm{tar}}_k,
\widehat F_{i^\ast}(\cdot)+\widehat q^{\mathrm{tar}}_k
\bigr].
$$

We assume the target-only predictor is also fitted consistently.

\begin{assumption}[Target-model prediction consistency]
\label{ass:target-consistency}
There exist sequences $\delta^{\mathrm{tar}}_{T,k}\ge 0$ and
$\rho^{\mathrm{tar}}_{T,k}\in[0,1]$ with $\delta^{\mathrm{tar}}_{T,k}\to 0$ and
$\rho^{\mathrm{tar}}_{T,k}\to 0$, and an event $\mathcal B^{\mathrm{tar}}_{T,k}$
with $\Pr(\mathcal B^{\mathrm{tar}}_{T,k})\ge 1-\rho^{\mathrm{tar}}_{T,k}$ such that
on $\mathcal B^{\mathrm{tar}}_{T,k}$,
\[
\sup_{\substack{(i,t)\in\mathcal D_{\mathrm{Cal}}\\ \widehat g(i)=k}}
\big|\widehat F_i(\cdot)-F_i(\cdot)\big|
\le
\delta^{\mathrm{tar}}_{T,k},
\qquad
\big|\widehat F_{i^\ast}(\cdot)-F_{i^\ast}(\cdot)\big|
\le
\delta^{\mathrm{tar}}_{T,k}.
\]
\end{assumption}

Now, we state our main theorem.

\begin{theorem}[Conformal interval length comparison]
\label{thm:interval-length}
Under Assumptions~\ref{ass:dnn-consistency}-\ref{ass:target-consistency}, the conformal cutoffs satisfy
$$
\widehat q^{\mathrm{joint}}_k
-
\widehat q^{\mathrm{tar}}_k
=
(\sigma_{e,k}-\sigma_{\epsilon,k})c_\alpha
+
O_p\left(
\varepsilon^{\mathrm{mix}}_{T,k}
+
\delta_{T,k}
+
\delta^{\mathrm{tar}}_{T,k}
+
\frac{1}{m_k}
\right),
$$
where $c_\alpha$ is the $(1-\alpha)$-quantile of $|Z|$ with
$Z\sim\mathcal N(0,1)$. Consequently, provided $\sigma_{\epsilon,k}>0$,
$$
\frac{
\mathrm{length}\big(\mathcal C^{\mathrm{joint}}_{i^\ast,T+1}(1-\alpha)\big)
}{
\mathrm{length}\big(\mathcal C^{\mathrm{tar}}_{i^\ast,T+1}(1-\alpha)\big)
}
=
\frac{\sigma_{e,k}}{\sigma_{\epsilon,k}}
+
O_p\left(
\varepsilon^{\mathrm{mix}}_{T,k}
+
\delta_{T,k}
+
\delta^{\mathrm{tar}}_{T,k}
+
\frac{1}{m_k}
\right).
$$
\end{theorem}
Theorem~\ref{thm:interval-length} provides a finite-calibration comparison of
conformal interval lengths. Even when the number of calibration points $m_k$
within a group is fixed, the joint modeling framework yields shorter prediction
intervals up to an explicit remainder term of order $1/m_k$, which reflects
irreducible calibration uncertainty. If $m_k\to\infty$ and the prediction and
mixing errors vanish, the ratio converges in probability to
$\sigma_{e,k}/\sigma_{\epsilon,k}\le 1$.

\begin{proof}[Proof of Theorem~\ref{thm:interval-length}]
For $m\in\{\mathrm{tar},\mathrm{joint}\}$, let $\widehat q_k^m$ denote the empirical
$(1-\alpha)$ conformal cutoff based on the $m_k$ calibration scores in group $k$.
By definition of the empirical quantile,
$$
1-\alpha
\le
\widetilde F_k^m(\widehat q_k^m)
\le
1-\alpha+\frac{1}{m_k}.
$$

Define
$$
\mathcal B_{T,k}^{\mathrm{joint}}:=\mathcal B_{T,k},
\qquad
\delta_{T,k}^{\mathrm{joint}}:=\delta_{T,k},
\qquad
\delta_{T,k}^{\mathrm{tar}}:=\delta_{T,k}^{\mathrm{tar}}.
$$
Consider the event $\mathcal B_{T,k}^m\cap \mathcal M_{T,k}$. By
Assumptions~\ref{ass:dnn-consistency}, \ref{ass:mixing}, and
\ref{ass:target-consistency}, we may apply the same argument as in
\eqref{equ:Jan23:01} to obtain
$$
\sup_{x\in\mathbb R}
\big|
\widetilde F_k^m(x)-F_k^m(x)
\big|
\le
3\varepsilon^{\mathrm{mix}}_{T,k}
+
4L_k\delta^m_{T,k}.
$$
Evaluating this bound at $x=\widehat q_k^m$ and combining with the empirical
quantile inequality above yields
\begin{equation}\label{equ:January3:02}
\big|
F_k^m(\widehat q_k^m)-(1-\alpha)
\big|
\le
3\varepsilon^{\mathrm{mix}}_{T,k}
+
4L_k\delta^m_{T,k}
+
\frac{1}{m_k}.
\end{equation}

Under Assumption~\ref{ass:gaussian}, the oracle scores $|\epsilon_{i,t}|$ and
$|e_{i,t}|$ follow half-normal distributions with scale parameters
$\sigma_{\epsilon,k}$ and $\sigma_{e,k}$, respectively. Hence their population
$(1-\alpha)$-quantiles are given by
$$
q_k^{\mathrm{tar}}=\sigma_{\epsilon,k}c_\alpha,
\qquad
q_k^{\mathrm{joint}}=\sigma_{e,k}c_\alpha,
$$
where $c_\alpha$ is the $(1-\alpha)$-quantile of $|Z|$ for
$Z\sim\mathcal N(0,1)$.

For each $m\in\{\mathrm{tar},\mathrm{joint}\}$, let $f_k^m$ denote the density of
the corresponding half-normal oracle score distribution. Since
$q_k^m>0$ whenever the scale parameter is positive, we have
$f_k^m(q_k^m)>0$. Therefore $F_k^m$ is continuously invertible in a
neighborhood of $q_k^m$, and the inverse map is locally Lipschitz there. It
follows from \eqref{equ:January3:02} that
$$
\widehat q_k^m-q_k^m
=
O_p\left(
\varepsilon^{\mathrm{mix}}_{T,k}
+
\delta_{T,k}^m
+
\frac{1}{m_k}
\right).
$$
Applying this with $m=\mathrm{joint}$ and $m=\mathrm{tar}$ gives
$$
\widehat q_k^{\mathrm{joint}}-q_k^{\mathrm{joint}}
=
O_p\left(
\varepsilon^{\mathrm{mix}}_{T,k}
+
\delta_{T,k}
+
\frac{1}{m_k}
\right),
$$
and
$$
\widehat q_k^{\mathrm{tar}}-q_k^{\mathrm{tar}}
=
O_p\left(
\varepsilon^{\mathrm{mix}}_{T,k}
+
\delta^{\mathrm{tar}}_{T,k}
+
\frac{1}{m_k}
\right).
$$
Therefore,
$$
\widehat q^{\mathrm{joint}}_k-\widehat q^{\mathrm{tar}}_k
=
(q^{\mathrm{joint}}_k-q^{\mathrm{tar}}_k)
+
O_p\left(
\varepsilon^{\mathrm{mix}}_{T,k}
+
\delta_{T,k}
+
\delta^{\mathrm{tar}}_{T,k}
+
\frac{1}{m_k}
\right),
$$
which yields the stated cutoff comparison.

Finally, since
$$
\mathrm{length}\big(\mathcal C^{\mathrm{joint}}_{i^\ast,T+1}(1-\alpha)\big)
=
2\widehat q_k^{\mathrm{joint}},
\qquad
\mathrm{length}\big(\mathcal C^{\mathrm{tar}}_{i^\ast,T+1}(1-\alpha)\big)
=
2\widehat q_k^{\mathrm{tar}},
$$
we have
$$
\frac{
\mathrm{length}\big(\mathcal C^{\mathrm{joint}}_{i^\ast,T+1}(1-\alpha)\big)
}{
\mathrm{length}\big(\mathcal C^{\mathrm{tar}}_{i^\ast,T+1}(1-\alpha)\big)
}
=
\frac{\widehat q_k^{\mathrm{joint}}}{\widehat q_k^{\mathrm{tar}}}.
$$
Because $\sigma_{\epsilon,k}>0$, we have
$q_k^{\mathrm{tar}}=\sigma_{\epsilon,k}c_\alpha>0$, so the denominator is bounded
away from zero with probability tending to one. Hence the preceding remainder
bounds imply
$$
\frac{\widehat q_k^{\mathrm{joint}}}{\widehat q_k^{\mathrm{tar}}}
=
\frac{q_k^{\mathrm{joint}}}{q_k^{\mathrm{tar}}}
+
O_p\left(
\varepsilon^{\mathrm{mix}}_{T,k}
+
\delta_{T,k}
+
\delta^{\mathrm{tar}}_{T,k}
+
\frac{1}{m_k}
\right)
=
\frac{\sigma_{e,k}}{\sigma_{\epsilon,k}}
+
O_p\left(
\varepsilon^{\mathrm{mix}}_{T,k}
+
\delta_{T,k}
+
\delta^{\mathrm{tar}}_{T,k}
+
\frac{1}{m_k}
\right),
$$
which proves the stated ratio comparison.
\end{proof}

\subsection{A General Size Control Theorem}
In this subsection, we provide a general assumption for the size control of Conformal prediction. The statement is motivated by \cite{xu2023conformal}.
Let $\widetilde F_k$ be the empirical CDF of the calibration scores
$\{s_{i,t}:(i,t)\in\mathcal D_{\mathrm{Cal}},\ \widehat g(i)=k\}$, and let $F_k$ be the
(population) CDF of the test score $s_{i^\ast,T+1}$ conditional on
$\widehat g(i^\ast)=k$. Let $ m_k := \big|\{(i,t)\in\mathcal D_{\mathrm{Cal}}:\widehat g(i)=k\}\big|$ be the sample size in the calibration set in the group $k$. We impose a uniform CDF approximation condition.
\begin{assumption}[Uniform CDF approximation of score]
\label{ass:cdf-uniform}
Assume that for each group $k$ there exist sequences $\varepsilon_{T,k}\ge 0$ and
$\rho_{T,k}\in[0,1]$ with $\varepsilon_{T,k}\to 0$ and $\rho_{T,k}\to 0$, and an event
$\mathcal A_{T,k}$ with $\Pr(\mathcal A_{T,k})\ge 1-\rho_{T,k}$ such that on
$\mathcal A_{T,k}$,
$$
\sup_{x\in\mathbb R}\big|\widetilde F_k(x)-F_k(x)\big|
\le
\varepsilon_{T,k}.
$$
\end{assumption}
This Assumption requires the empirical score CDF converges to population score CDF, which directly postulates that the calibration-score empirical CDF approximates the conditional CDF of the test score

The following lemma gives the asymptotic conditional size control based on this condition. 
\begin{lemma}[Asymptotic conditional size control]
\label{lem:conditional-size}
Under Assumption~\ref{ass:cdf-uniform}, on the event $\mathcal A_{T,k}$,
$$
\bigg|
\Pr\left(
y_{i^\ast,t+h}\in \mathcal C_{i^\ast,t+h}(1-\alpha)
\ \middle|\ \widehat g(i^\ast)=k
\right)
-(1-\alpha)
\bigg|
\le
\varepsilon_{T,k}+\frac{1}{m_k}.
$$
\end{lemma}

\begin{proof}[Proof of Lemma \ref{lem:conditional-size}]
The proof of lemma is similar to Theorem 1 of \cite{xu2023conformal}, we include the proof here for completeness.
Fix a group $k$ and a test point $(i^\ast,t+h)$ with $\widehat g(i^\ast)=k$.  Let
$$
S^\ast:=s_{i^\ast,t+h}
=
\big|y_{i^\ast,t+h}-\widehat y_{i^\ast,t+h}\big|
$$
denote the test score, and let $F_k$ be its conditional CDF given
$\widehat g(i^\ast)=k$. By construction of the interval,
$$
\bigl\{y_{i^\ast,t+h}\in \mathcal C_{i^\ast,t+h}(1-\alpha)\bigr\}
\iff
\bigl\{|y_{i^\ast,t+h}-\widehat y_{i^\ast,t+h}|\le q_k\bigr\}
\iff
\{S^\ast\le q_k\}.
$$
Therefore,
$$
\Pr\left(
y_{i^\ast,t+h}\in \mathcal C_{i^\ast,t+h}(1-\alpha)
\ \middle|\ \widehat g(i^\ast)=k
\right)
=
F_k(q_k).
$$

By definition of the left empirical quantile,
$$
\widetilde F_k(q_k^-)
<
1-\alpha
\le
\widetilde F_k(q_k),
$$
where $\widetilde F_k(q_k^-)=\lim_{x\uparrow q_k}\widetilde F_k(x)$ is the left limit of $\tilde{F}_k$ at $q_k$.
Since $\widetilde F_k$ increases in jumps of size at most $1/m_k$. Therefore
$$
\widetilde F_k(q_k) - (1-\alpha) \le \widetilde F_k(q_k) - \widetilde F_k(q_k^{-}) \le \frac{1}{m_k}
$$
Then,
$$
0
\le
\widetilde F_k(q_k)-(1-\alpha)
\le
\frac{1}{m_k}.
$$
Hence,
\begin{equation}\label{equ:dec17:01}
\big|\widetilde F_k(q_k)-(1-\alpha)\big|
\le
\frac{1}{m_k}.
\end{equation}

On the event $\mathcal A_{T,k}$, Assumption~\ref{ass:cdf-uniform} implies
\begin{equation}\label{equ:dec17:02}
  \big|F_k(q_k)-\widetilde F_k(q_k)\big|
\le
\varepsilon_{T,k}.  
\end{equation}
Combining \eqref{equ:dec17:01} and \eqref{equ:dec17:02} yields, on $\mathcal A_{T,k}$,
$$
\big|F_k(q_k)-(1-\alpha)\big|
\le
\varepsilon_{T,k}+\frac{1}{m_k}.
$$
\end{proof}


\section{Shortcut Learning}
\label{app:shortcut}

In this section, we formalize the shortcut learning phenomenon
\citep{geirhos2020shortcut} that may arise when a surrogate response is directly
incorporated as a predictor for the target outcome. This phenomenon provides a
key motivation for residual modeling, rather than directly regressing the
target on the surrogate output.

To isolate the main intuition, we study a simplified setting. For clarity, we
suppress the cross-sectional index $i$ and the text index $k$. Suppose the
target outcome satisfies
$$
y = F(\bx,\bz) + \epsilon,
$$
while the surrogate response is generated by
$$
y^S = G(\bx) + \epsilon^S.
$$
Thus, the surrogate captures a component driven only by $\bx$, whereas the
target depends on the richer signal $(\bx,\bz)$.

We first consider the \emph{direct modeling} strategy, in which the surrogate
response is inserted directly into the predictor:
$$
\widehat y
=
f_\theta(\bx,\bz) + ay^S,
$$
where $f_\theta$ is a deep network indexed by parameters $\theta$, and $a$ is a
scalar controlling the surrogate channel. The empirical squared-error loss is
$$
L(\theta,a)
=
\frac{1}{2T}\sum_{t=1}^T
\Bigl(y_t-f_\theta(\bx_t,\bz_t)-ay_t^S\Bigr)^2.
$$
Let
$$
r_t(\theta,a)
:=
y_t-f_\theta(\bx_t,\bz_t)-ay_t^S
$$
denote the residual. Then the gradients of the loss are
$$
\frac{\partial L}{\partial a}(\theta,a)
=
-\frac{1}{T}\sum_{t=1}^T r_t(\theta,a)y_t^S,
\qquad
\frac{\partial L}{\partial \theta}
=
-\frac{1}{T}\sum_{t=1}^T r_t(\theta,a)
\frac{\partial f_\theta(\bx_t,\bz_t)}{\partial \theta}.
$$

A key issue is that the direct inclusion of $y^S$ may create a shortcut
channel whose gradient dominates the early stage of training. To see this,
consider a random initialization $(\theta_0,a_0)$ such that
$f_{\theta_0}(\bx_t,\bz_t)\approx 0$ and $a_0\approx 0$. Then
$$
r_t(\theta_0,a_0)\approx y_t,
$$
so
\begin{equation}
\label{eq:shortcut-grad-a}
\frac{\partial L}{\partial a}(\theta_0,a_0)
=
-\frac{1}{T}\sum_{t=1}^T y_ty_t^S.
\end{equation}
If $y_t$ and $y_t^S$ are centered, the right-hand side is approximately
$-\Cov(y,y^S)$, which can be large whenever the surrogate is strongly
predictive of the target.

By contrast, for a randomly initialized network, the coordinate-wise quantities
$\partial f_{\theta_0}(\bx_t,\bz_t)/\partial \theta_j$ typically have weak
alignment with $y_t$. Hence each individual network gradient component
\begin{equation}
\label{eq:shortcut-grad-theta}
\frac{\partial L}{\partial \theta_j}(\theta_0,a_0)
=
-\frac{1}{T}\sum_{t=1}^T
y_t\frac{\partial f_{\theta_0}(\bx_t,\bz_t)}{\partial \theta_j}
\end{equation}
is often much smaller in magnitude than the shortcut gradient in
\eqref{eq:shortcut-grad-a}. This creates an \emph{initial gradient imbalance}:
the optimization dynamics can reduce the loss much faster by fitting the
one-dimensional surrogate channel than by coordinating the many parameters of
the deep network.

Consequently, in the early stage of gradient descent, the update in $a$ tends
to be much larger than the update of a typical coordinate of $\theta$. The
optimization path is therefore biased toward first fitting the easy linear
dependence of $y$ on $y^S$, whose population least-squares coefficient is
$$
a^\star
=
\frac{\Cov(y,y^S)}{\Var(y^S)}
$$
when the variables are centered. After only a few iterations, the predictor may
already be well approximated by
$$
\widehat y \approx a^\star y^S,
$$
so that a substantial fraction of the explainable variation in $y$ has been
absorbed by the surrogate channel.

At that point, the residuals
$r_t(\theta,a)\!=\!y_t-f_\theta(\bx_t,\bz_t)-ay_t^S$ can become much smaller in
magnitude, and the remaining signal available to train the deep network may
have a substantially lower signal-to-noise ratio. As a result, the gradient
signal for learning the richer structural component $F(\bx,\bz)$ becomes much
weaker:
$$
\frac{\partial L}{\partial \theta_j}(\theta,a)
=
-\frac{1}{T}\sum_{t=1}^T
r_t(\theta,a)
\frac{\partial f_\theta(\bx_t,\bz_t)}{\partial \theta_j}.
$$
Thus, although the hypothesis class of $f_\theta$ is highly expressive, the
optimization trajectory may become trapped near a shallow,
surrogate-dominated solution of the form
$$
\widehat y \approx a^\star y^S.
$$
This is the shortcut dominance phenomenon: optimization preferentially exploits
the easiest predictive channel, preventing the network from fully learning the
structural signal encoded in $F(\bx,\bz)$.

We now contrast this with \emph{residual modeling}. Define the surrogate
residual
$$
\epsilon^S := y^S-G(\bx),
$$
and consider the residual-augmented target model
$$
\widehat y
=
f_\theta(\bx,\bz)+b\epsilon^S.
$$
At initialization, by the same reasoning as in
\eqref{eq:shortcut-grad-a},
$$
\frac{\partial L}{\partial b}(\theta_0,b_0)
\approx
-\frac{1}{T}\sum_{t=1}^T y_t\epsilon_t^S.
$$
If the dominant predictive content of $y^S$ comes from the common
$\bx$-driven component $G(\bx)$, then
$$
\Big|\Cov(y,\epsilon^S)\Big|
\ll
\Big|\Cov(y,y^S)\Big|,
$$
again under centering. In this case, the residualized surrogate no longer
produces a large shortcut gradient at initialization. Gradient descent is then
less attracted to the shortcut channel and must reduce the loss primarily by
improving the structural predictor $f_\theta(\bx,\bz)$.

Accordingly, the residual term $b\epsilon^S$ acts only as a secondary,
variance-reducing correction, rather than as a dominant explanatory variable.
Residualization therefore mitigates shortcut dominance, improves the
recoverability of the structural target function $F(\bx,\bz)$, and makes the
subsequent homogeneity-pursuit step less sensitive to surrogate-driven
artifacts. In this sense, residual modeling helps ensure that grouping is more
closely aligned with the structural macro--text response, rather than with the
easy predictive shortcut carried by the raw surrogate output.

\end{CJK}
\end{document}